\documentclass[11pt]{article}
\usepackage{amsfonts}
\usepackage{amsmath}
\usepackage{amssymb}
\usepackage{graphicx}
\usepackage{geometry, tabularx}
\providecommand{\U}[1]{\protect\rule{.1in}{.1in}}

\newtheorem{thm}{Theorem}

\newtheorem{cor}{Corollary}

\newtheorem{rem}{Remark}

\usepackage{graphicx,color}
\usepackage[bf,SL,BF]{subfigure}

\newenvironment{proof}[1][Proof]{\noindent\textbf{#1.} }{\ \rule{0.5em}{0.5em}}
\normalsize
\numberwithin{equation}{section}
\allowdisplaybreaks

\begin{document}

\title{Risk Modelling on Liquidations with L\'{e}vy Processes
}

\author{Aili Zhang\thanks{School of Statistics and Mathematics, Nanjing Audit University, Nanjing 211815, China. E-mail: zhangailiwh@126.com.},\,\,\,\,\, Ping Chen\thanks{Department of Economics, The University of Melbourne, Parkville,
Victoria 3010, Australia. E-mail: pche@unimelb.edu.au},\,\,\,\,\, Shuanming Li\thanks{Department of Economics, The University of Melbourne, Parkville,
Victoria 3010, Australia. E-mail: shli@unimelb.edu.au}\,\,\,\,\, and\,\,\,\,\,Wenyuan Wang\thanks{Corresponding author. School of Mathematical Sciences, Xiamen University, Xiamen, Fujian, China. E-mail: wwywang@xmu.edu.cn}}

\maketitle
\vspace{-.25in}
\begin{abstract}
It has been decades since the academic world of ruin theory defined the insolvency of an insurance company as the time when its surplus falls below zero. This simplification, however, needs careful adaptions to imitate the real-world liquidation process. Inspired by Broadie et al. (2007) and Li et al. (2020), this paper uses a three-barrier model to describe the financial stress towards bankruptcy of an insurance company. The financial status of the insurer is divided into solvent, insolvent and liquidated three states, where the insurer's surplus process at the state of solvent and insolvent is modelled by two spectrally negative L\'{e}vy processes, which have been taken as good candidates to model insurance risks. We provide a rigorous definition of the time of liquidation ruin in this three-barrier model. By adopting the techniques of excursions in the fluctuation theory, we study the joint distribution of the time of liquidation, the surplus at liquidation and the historical high of the surplus until liquidation, which generalizes the known results on the classical expected discounted penalty function in Gerber and Shiu (1998). The results have semi-explicit expressions in terms of the scale functions and the L\'{e}vy triplets associated with the two underlying L\'{e}vy processes. The special case when the two underlying L\'{e}vy processes coincide with each other is also studied, where our results are expressed compactly via only the scale functions. The corresponding results have good consistency with the existing literatures on Parisian ruin with (or without) a lower barrier in Landriault et al. (2014), Baurdoux et al. (2016) and Frostig and Keren-Pinhasik (2019).
Besides, numerical examples are provided to illustrate the underlying features of liquidation ruin.

\end{abstract}\vspace{-.15in}


\vspace{-.2in}

\section{Introduction}\vspace{-.2in}

In the classical risk theory, the ultimate ruin occurs if an insurance company has no sufficient assets to meet its liabilities. That is, when the surplus ever falls to zero or below, the insurance company is defined as insolvent. In this somewhat vague statement, the state of solvency is described by a single barrier zero of the surplus process. Under this definition, ruin theory has been an area of study for actuaries and mathematicians for many decades. However, the practical issue of solvency is far more complicated than this approximation.

In the United States, bankruptcy is governed by the \lq\lq US Bankruptcy Code\rq\rq. 
To resolve the financial distress of a business whose liability is greater than its assets, two options can be taken when filing for bankruptcy: Chapter 7 and Chapter 11. A business can seek Chapter 7 bankruptcy only if it immediately ceases all operations and hence Chapter 7 is usually referred to as \lq\lq liquidation", this type of bankruptcy is a popular option for small business owners. Chapter 11 bankruptcy, on the other hand, is often referred to as \lq\lq reorganization" or \lq\lq rehabilitation" bankruptcy because the process provides the business with an opportunity to reorganize its financial structure such as debts and to try to re-emerge as a healthy organization while continuing to operate. Chapter 11, which is more expensive than Chapter 7, is typically intended for medium to large-sized businesses. According to Corbae and D'Erasmo (2017), 80\% of bankruptcies of publicly traded companies  are filed under Chapter 11, while only 20\% are Chapter 7 liquidations.




From academic point of view, it is natural and straightforward to describe Chapter 7 bankruptcy by a single barrier model of the surplus process, see for example Leland (1994), where bankruptcy is triggered when the firm's asset value falls to the debt's principal value. While Chapter 11 is a long and costly process, the grant of a chance to reorganize debts has been an attractive feature to many large corporations. Therefore, more and more researchers in finance try to imitate the process of Chapter 11 bankruptcy in their models. Broadie et al. (2007) adopted a three-barrier model to characterize the financial states of a bank: distress, Chapter 11 bankruptcy and Chapter 7 liquidation. Antill and Grenadier (2019) formulated a model in which shareholders can choose both their timing of default and the chapter of bankruptcy. They also examined how this flexibility alters the capital structure decisions of the firm.

Given the significance of Chapter 11 filing, the regulators of insurance companies also keep up with the trend aiming to provide crucial safeguards for policyholders and for the economy. A review on the changes in the U.S. insurance regulatory system can be found in Li et al. (2020). We only give a brief introduction on the current regulatory system in the U.S. and E.U. for self-completeness.

In the United States, insurance companies are not subject to the federal bankruptcy code, and their financial health is monitored by the National Association of Insurance Commissioners (NAIC) and vary by state. The primary goal is to make sure there is sufficient capital for insurers to operate and meet their obligations to policyholders and other claimants. The U.S. method of measuring whether capital is adequate is called Risk Based Capital (RBC). Based on the amount of insurance the company writes, the lines of business it writes, the assets it invests in and other measures, an absolute least amount of capital an insurer needs is calculated. If an insurer's capital dwindles, the regulators have the opportunity to intercede. The closer the actual capital gets to the risk-based minimum, the more powerful the intervention can be.

Analogous to the regulatory spirit in the U.S., the current supervisory regime for insurance companies in the European Union is also a risk-based capital regime. Since January 2016, the E.U. insurers are governed by the Solvency II regulatory regime. A three-level capital requirement system is applied: technical provision is to fulfil the obligations towards policyholders and other beneficiaries; minimum capital requirement (MCR) is a safety net and reflects a level of capital below which ultimate supervisory action would be triggered; solvency capital requirement (SCR) enables an insurance company to absorb significant unforeseen losses and that gives reasonable assurance to policyholders and beneficiaries.

To imitate this real-world process of bankruptcy, the classical definition of ruin describes Chapter 7 liquidation reasonably well by using a single barrier, but fails to describe the complex Chapter 11 reorganization. In view of the recent regulatory development described in the above, and following the pioneering work by Li et al. (2020) in the insurance sector, we use a three-barrier model to describe the financial stress towards bankruptcy of an insurance company. This paper distinguishes between {\it insolvent} and {\it bankruptcy}. A company can be insolvent without being bankrupt, but it cannot be bankrupt without being insolvent. Insolvency is the inability to pay debts when they are due, while bankruptcy is usually a final alternative with the failures of all the other attempts to clear debt.

To better reflect the meaning of the three barriers, this paper adopts a different naming system of the three barriers from Li et al. (2020). The highest barrier is called \lq\lq safety\rq\rq barrier (denoted by $c$). If the surplus process stays above this barrier, this indicates the insurer has enough buffer to settle all claims in extreme situations which means the insurer is a healthy financial institution. Otherwise, an investigation of the insurer's business is carried out by the regulator on the insurer's financial capacity to meet both its short-term and long-term liabilities. If the surplus process keeps shrinking, the intermediary barrier is called \lq\lq reorganization\rq\rq or \lq\lq rehabilitation\rq\rq barrier (denoted by $b$), which triggers the regulator's intervention in the insurer's business operation. The regulator may give a broad range of directions to the insurer with respect to the carrying on of its business, including prohibiting the insurer from issuing further policies, prohibiting it from disposing of assets, requiring it to make provisions in its accounts and requiring it to increase its paid-up capital. The regulator's rehabilitation intervention continues until the surplus climbs up to the safety barrier, or transfers to the wind-up procedure when the lowest barrier is breached. We call the lowest barrier \lq\lq liquidation\rq\rq barrier (denoted by $a$). Once an insurance company has been liquidated, it is completely dissolved and permanently ceases operations. The three-barrier system with $a<b<c$ divides the state of an insurer into {\it solvent}, {\it insolvent} and {\it liquidated} ({\it bankrupt}) three states. The detailed transition between states is presented in Section 2.

Another feature of this paper is the insurer¡¯s surplus process is modelled by spectrally negative L\'{e}vy processes, which are stochastic processes with stationary independent increments and with sample paths of no positive jumps. L\'{e}vy processes have been taken as good candidates to model insurance risk, see Shimizu and Zhang (2019) for a review for these features of L\'{e}vy processes. The application of spectrally negative L\'{e}vy processes in risk theory can be seen in Yang and Zhang (2001), Huzak et al. (2004), Chiu and Yin (2005), Garrido and Morales (2006), Biffis and Morales (2010), Cheung et al. (2010), Yin et al. (2014), Albrecher et al. (2016), 
Loeffen et al. (2018), Wang et al. (2020) and Wang and Zhou (2020). Based on the time of liquidation, this paper studies an extended definition of the expected discounted penalty function expresses in terms of the q-scale functions and the L\'{e}vy triplet associated with the L\'{e}vy process. The joint distribution of the time of liquidation, the surplus at liquidation and the historical high of the surplus until liquidation is derived.

From technical point of view, our mathematical argument is based on a heuristic idea presented in Li et al. (2020) which consists of distinguishing the excursions away from the three barriers (i.e., $a$, $b$ and $c$) of the underlying diffusion surplus process. In the context of L\'evy processes, we  provide a rigorous definition of the time of liquidation ruin. With the help of the fluctuation theory for spectrally negative L\'evy processes (particularly in Loeffen et al. (2014), Wang et al. (2020), etc.) and delicate path analysis arguments, this paper derives a semi-explicit and compact expression of the extended Gerber-Shiu function at liquidation ruin, in terms of the so-called scale functions and the L\'evy triplet associated with the underlying L\'evy process. From our results, we can easily deduce the Gerber-Shiu distribution and the two-sided exit identities
at Parisian ruin which was originally obtained by Landriault et al. (2014), Baurdoux et al. (2016), and Frostig and Keren-Pinhasik (2019). In addition, compared with the fluctuation identities obtained in Landriault et al. (2014) and Frostig and Keren-Pinhasik (2019) where the spectrally negative L\'evy process is assumed to have bounded path variation, this paper unifies the results for bounded and unbounded variation.


The remaining part of this paper is organized as follows. In Section 2 we review the basics of
the spectrally negative L\'{e}vy processes, the associated scale functions, and the existing results of the exit problems. Section 3 presents the semi-explicit expression of the extended Gerber-Shiu function at the time of liquidation ruin. The application of our main results in the case of Parisian ruin is provided in Section 4. In Section 5, several numerical examples are studied to illustrate the features of our results.

 \vspace{-.2in}


\section{Preliminaries of the Spectrally Negative L\'evy Process}
\setcounter{section}{2}


In this section we review some preliminaries and fundamental results for fluctuation problems of the spectrally negative L\'{e}vy processes.
Let $X=\{X_{t};t\geq0\}$ be a spectrally negative L\'{e}vy process defined on a filtered probability space $(\Omega,\{\mathcal{F}_{t};t\geq0\},\mathbb{P})$ with the natural filtration $\{\mathcal{F}_{t};t\geq0\}$, that is, $X$ is a stochastic process with stationary and independent increments and no
positive jumps. To avoid trivialities, we exclude the case where $X$ has monotone paths.
Denote by $\mathbb{P}_{x}$ the conditional probability given $X_{0}=x$, and by $\mathbb{E}_{x}$ the associated conditional expectation. For notational convenience, we write $\mathbb{P}$ and $\mathbb{E}$ in place of $\mathbb{P}_{0}$ and $\mathbb{E}_{0}$, respectively. The Laplace transform of spectrally negative L\'{e}vy process $X$ is given by
$$\mathbb{E}\left(\mathrm{e}^{\theta X_{t}}\right)=\mathrm{e}^{t\psi(\theta)},$$ for all $\theta\geq 0$, where
$$\psi(\theta)=\gamma\theta+\frac{1}{2}\sigma^2\theta^2+\int_{0}^{\infty}
(\mathrm{e}^{\theta z}-1+\theta z\mathbf{1}_{(0,1]}(z))\upsilon(\mathrm{d}z),$$
for $\gamma\in(-\infty,\infty)$ and $\sigma\geq 0$, where $\upsilon$ is a $\sigma$-finite measure on $(0,\infty)$ such that
$$\int_{0}^{\infty}(1\wedge z^2)\upsilon(\mathrm{d}z)<\infty.$$
The measure $\upsilon$ is called the L\'{e}vy measure of $X$, and $(\gamma,\sigma,\upsilon)$ is called the L\'{e}vy triplet of $X$. Since the Laplace exponent $\psi$ is strictly convex and $\lim_{\theta\rightarrow \infty}\psi(\theta)=\infty$, there exists the right inverse of $\psi$ defined by $\Phi_{q}=\sup\{\theta\geq 0:\psi(\theta)=q\}$. We present the definition of the scale functions $W_{q}$ and $Z_{q}$ of $X$. For $q\geq0$, $W_{q}:\,[0,\infty)\rightarrow[0,\infty)$ is defined as the unique strictly increasing and continuous function on $[0,\infty)$ with Laplace transform
\begin{eqnarray}
\int_{0}^{\infty}\mathrm{e}^{-\theta x}W_{q}(x)\mathrm{d}x=\frac{1}{\psi(\theta)-q},\quad \theta>\Phi_{q}.\nonumber
\end{eqnarray}
For convenience, we extend the domain of $W_{q}(x)$ to the whole real line by setting $W_{q}(x)=0$ for $x<0$. Associated to the first scale function $W_{q}$ is the second scale function $Z_{q}$ defined by
\begin{eqnarray}
Z_{q}(x)=1+q\int_{0}^{x}W_{q}(z)\mathrm{d}z,\,\,\,\,x\in[0,\infty),\nonumber
\end{eqnarray}
with $Z_{q}(x)\equiv1$ over the negative half line $(-\infty,0)$.
For notional simplicity, we write $W=W_{0}$ and $Z=Z_{0}$.
In addition, for $p,q\geq 0$, $x>0$ and $x+w\geq 0$, we define
\begin{eqnarray}
\omega^{(q,p)}(w,x)\hspace{-0.3cm}&:=&\hspace{-0.3cm}W_{p}(w+x)-(p-q)\int_{0}^{x}W_{p}(z+w)W_{q}(x-z)\mathrm{d}z
\nonumber\\
\hspace{-0.3cm}&=&\hspace{-0.3cm}
W_{q}(x+w)+(p-q)\int_{0}^{w}W_{q}(x+w-z)W_{p}(z)\mathrm{d}z,\label{equi.def.w}
\end{eqnarray}
and
\begin{eqnarray}
\ell^{(q,p)}(w,x)\hspace{-0.3cm}&:=&\hspace{-0.3cm}Z_{p}(w+x)-(p-q)\int_{0}^{x}Z_{p}(z+w)W_{q}(x-z)\mathrm{d}z
\nonumber\\
\hspace{-0.3cm}&=&\hspace{-0.3cm}
Z_{q}(x+w)+(p-q)\int_{0}^{w}W_{q}(x+w-z)Z_{p}(z)\mathrm{d}z.\label{equi.def.z}
\end{eqnarray}
It is well known that
\begin{eqnarray}\label{s.f.l.}
\lim_{x\rightarrow\infty}\frac{W^{(q)\prime}(x)}{W_{q}(x)}=\Phi_{q} ,\quad \lim_{y\rightarrow\infty}\frac{W_{q}(x+y)}{W_{q}(y)}=\mathrm{e}^{\Phi_{q}x}.
\end{eqnarray}

The first passage times of level $w\in (-\infty,\infty)$ for the process $X$ is defined as
\begin{eqnarray}
\tau_{w}^{+}=\inf\{t\geq 0: X_{t}>w\} \,\,\, {\rm and} \,\,\, \tau_{w}^{-}=\inf\{t\geq 0: X_{t}<w\},\nonumber
\end{eqnarray}
with the convention that $\inf\emptyset=\infty$.
For $w>0$ and $x\leq w$, the two-sided exit problem is given by
\begin{eqnarray}\label{two.side.e.}
\mathbb{E}_{x}\left(\mathrm{e}^{-q\tau_{w}^{+}}\mathbf{1}_{\{\tau_{w}^{+}<\tau_{0}^{-}\}}\right)
=\frac{W_{q}(x)}{W_{q}(w)}, \quad  x\in(-\infty, w],
\end{eqnarray}
\begin{eqnarray}\label{two.side.d.}
\mathbb{E}_{x}\left(\mathrm{e}^{-q\tau_{0}^{-}}\mathbf{1}_{\{\tau_{0}^{-}<\tau_{w}^{+}\}}\right)
=Z_{q}(x)-Z_{q}(w)\frac{W_{q}(x)}{W_{q}(w)}, \quad  x\in(-\infty, w].
\end{eqnarray}
We recall the resolvent measure as follows.
For $q\geq 0, w>0$, and $x,y\in[0,w]$,
\begin{eqnarray}\label{r.m.}
\int_{0}^{\infty}\mathrm{e}^{-qt}\mathbb{P}_{x}(X_{t}\in\mathrm{d}y, t<\tau_{w}^{+}\wedge\tau_{0}^{-})\mathrm{d}t=\left(\frac{W_{q}(x)W_{q}(w-y)}
{W_{q}(w)}-W_{q}(x-y)\right)\mathrm{d}y.
\end{eqnarray}
It is seen from Corollary 3.2 and Remark 3.3 in Wang et al. (2020) that
\begin{eqnarray}\label{mar.}
\hspace{-0.3cm}&&\hspace{-0.3cm}
\mathbb{E}_{x}\left(\mathrm{e}^{-q\tau_{w}^{-}}; \,X(\tau_{w}^{-}-)\in\mathrm{d}y,
\,-X(\tau_{w}^{-})\in\mathrm{d}\theta,\,\tau_{w}^{-}<\tau_{z}^{+}\right)
\nonumber\\
\hspace{-0.3cm}&=&
\hspace{-0.3cm}
W_{q}(x-w)\,\upsilon\left(y+\mathrm{d}\theta\right)
\left[
\frac{W_{q}(z-y)}{W_{q}(z-w)}-\frac{W_{q}(x-y)}{W_{q}(x-w)}
\right]\mathbf{1}_{\{w<y\leq z\}} \mathbf{1}_{\{-w<\theta\}} \mathrm{d}y,\quad w\leq x\leq z,
\end{eqnarray}
and
\begin{eqnarray}\label{mar.1}
\hspace{-0.3cm}
\mathbb{E}_{x}\left[\mathrm{e}^{-q\tau_{w}^{-}}; \,X(\tau_{w}^{-})=w,
\,\tau_{w}^{-}< \tau_{z}^{+}\right]
=
\frac{\sigma^{2}}{2}\left[W_{q}^{\prime}(x-w)-W_{q}(x-w)\frac{W_{q}^{\prime}(z-w)}
{W_{q}(z-w)}\right],\quad w\leq x\leq z,
\end{eqnarray}
where $\sigma$ is the Gaussian coefficient associated with the L\'evy process $X$.

The three barriers in our model are denoted by $a$ (liquidation or bankruptcy), $b$ (rehabilitation or reorganization) and $c$ (safety). We require $a<b<c$ and $c>0$. The state of liquidation is an absorbing state, that is, once an insurer's surplus process hit $a$, it is completely dissolved and permanently ceases operations. When the surplus stays above $c$, the insurer is a healthy financial institution and is able to meet all its liabilities. Once the surplus process falls below $c$, the insurer is under regulator's supervision, which does not lead to the intervention of the insurer's daily operations but an investigation on the insurer's financial status will be conducted. As the situation is getting worse, once $b$ is down crossed, the regulators intervene the daily operation and hence the surplus process will follow a different spectrally negative L\'{e}vy process as given in \eqref{model} in below.

The three-barrier system with $a<b<c$ divides the state of an insurer into {\it solvent}, {\it insolvent} and {\it liquidated} ({\it bankrupt}) three states. Before the final liquidation, the insurer remains at the state of {\it solvent} or {\it insolvent}, where different spectrally negative L\'{e}vy processes are adopted, indicating that the regulators can intervene the everyday operations of the insurer and hence can alter the dynamics of the surplus process. To measure the maximum tolerable time period for the regulator to allow the insurer as staying in the state of {\it insolvent}, a \lq\lq grace period\rq\rq is granted. Then the state of {\it insolvent} refers to a period of time in which the surplus starts from a value equals to or below $b$ until the lowest barrier $a$ is hit and transfers directly to the state of {\it liquidated}, or until the up-crossing time of barrier $c$ within the grace period and transfers to the state of {\it solvent}, or until the staying period between $a$ and $c$ exceeds the granted grace period and transfers to the state of {\it liquidated}. In other words, a surplus value below $b$ indicates a poor financial position of the insurer, the state of which can only be transferred to {\it solvent} when the safety barrier $c$ can be achieved, within the tolerable grace period, otherwise, the insurer will be liquidated.

In this paper, the surplus process $U$ will be given by
\begin{eqnarray}
\label{model}
\mathrm{d}U_{t}
=\begin{cases}
\mathrm{d}X_{t},\quad\mbox{when } U_{t}\mbox{ is in solvent state},\\
\mathrm{d}\widetilde{X}_{t},\quad\mbox{when } U_{t}\mbox{ is in insolvent state},
\end{cases}
\end{eqnarray}
with $\widetilde{X}$ being another L\'evy process with L\'evy triplet
$(\widetilde{\gamma},\widetilde{\sigma},\widetilde{\upsilon})$ and Laplace exponent $\widetilde{\psi}$. Let the two scale functions $\mathbb{W}_{q}$ and $\mathbb{Z}_{q}$ associated with $\widetilde{X}$ be defined in the same manner as $W_{q}$ and $Z_{q}$ while with $\psi$ replaced by $\widetilde{\psi}$, and write $\mathbb{W}=\mathbb{W}_{0}$ and $\mathbb{Z}=\mathbb{Z}_{0}$.
In addition, denote by $\widetilde{\mathbb{P}}_{x}$ the conditional probability given $\widetilde{X}_{0}=x$, and by $\widetilde{\mathbb{E}}_{x}$ the associated conditional expectation. 
And, with a little abuse of notations, under $\widetilde{\mathbb{P}}_{x}$, $\tau_{w}^{+}$ ($\tau_{w}^{-}$) is taken as
the first up-crossing (down-crossing ) time of $w$ for the process $\widetilde{X}$.

Define
$$\zeta^{-}_{b,1}=\zeta^{-}_{b}:=\inf\{t\geq 0; U_{t}<b\}\,\,\,\mbox{ and }\,\,\,\zeta^{+}_{c,1}=\zeta^{+}_{c}:=\inf\{t\geq \zeta^{-}_{b,1}; U_{t}>c\},$$
respectively to be the first time that the process $U$ down-crosses the level $b$ and the first time $U$ up-crosses the level $c$ after $\zeta^{-}_{b,1}$. We recursively define two sequences of stopping times $\{\zeta^{-}_{b,k}\}_{k\geq 1}$ and $\{\zeta^{+}_{c,k}\}_{k\geq 1}$ as follows. For $k\geq 2$, let $$\zeta^{-}_{b,k}:=\inf\{t\geq \zeta^{+}_{c,k-1}; U_{t}<b\}\,\,\, \mbox{ and }\,\,\,\zeta^{+}_{c,k}:=\inf\{t\geq \zeta^{-}_{b,k}; U_{t}>c\}.$$
Then, we define $$\kappa:=\inf\{k\geq 1; \zeta^{-}_{b,k}+e_{\lambda}^{k}<\zeta^{+}_{c,k}\}$$ and let
\begin{eqnarray}
\zeta_{b, c, e_{\lambda}}:=\begin{cases}
\zeta^{-}_{b,\kappa}+e_{\lambda}^{\kappa}\,\,\,\,\,\,\mbox{ when }\,\,\,\kappa<\infty,\\
\infty,\quad\quad\quad\,\, \mbox{ when }\,\,\,\kappa=\infty,
\end{cases}\nonumber
\end{eqnarray}
where $\{e_{\lambda}^{k}\}_{k\geq 1}$ is a sequence of independent and  exponentially distributed random variables with parameter $\lambda$. We use $\{e_{\lambda}^{k}\}_{k\geq 1}$ to model the duration of grace period granted by the regulator. We also assume that $\{e_{\lambda}^{k}\}_{k\geq 1}$ is independent of $X$ and $\widetilde{X}$.

\begin{rem}
\label{grace}
A period of grace is the maximal amount of time the company is allowed to stay in the state of {\it insolvent}. It is offered to enable the company to better shape its financial condition. In practice it could range from several months to several years, see Broadie et al. (2007). The grace period could be thought of as a multiple renegotiated period, according to the creditors' interest and the regulators' policies. And hence is usually not a deterministic constant. As adopted by Li et al. (2020), we use a sequence of independent and exponentially distributed random variables to model the grace periods when the insurer is in the state of {\it insolvent}.
\end{rem}

Based on the passage times defined in the above, we introduce the time of liquidation as
\begin{eqnarray}
\label{liquidation}
T:=\zeta_{a}^{-}\wedge \zeta_{b, c, e_{\lambda}}.
\end{eqnarray}
To avoid the trivial case that liquidation occurs with probability one, we assume throughout the paper that $\psi^{\prime}(0+)>0$, that is, the positive safety loading condition holds.

\begin{figure}[!htp]
\centering{}\includegraphics[width=6.5in,height=4in]{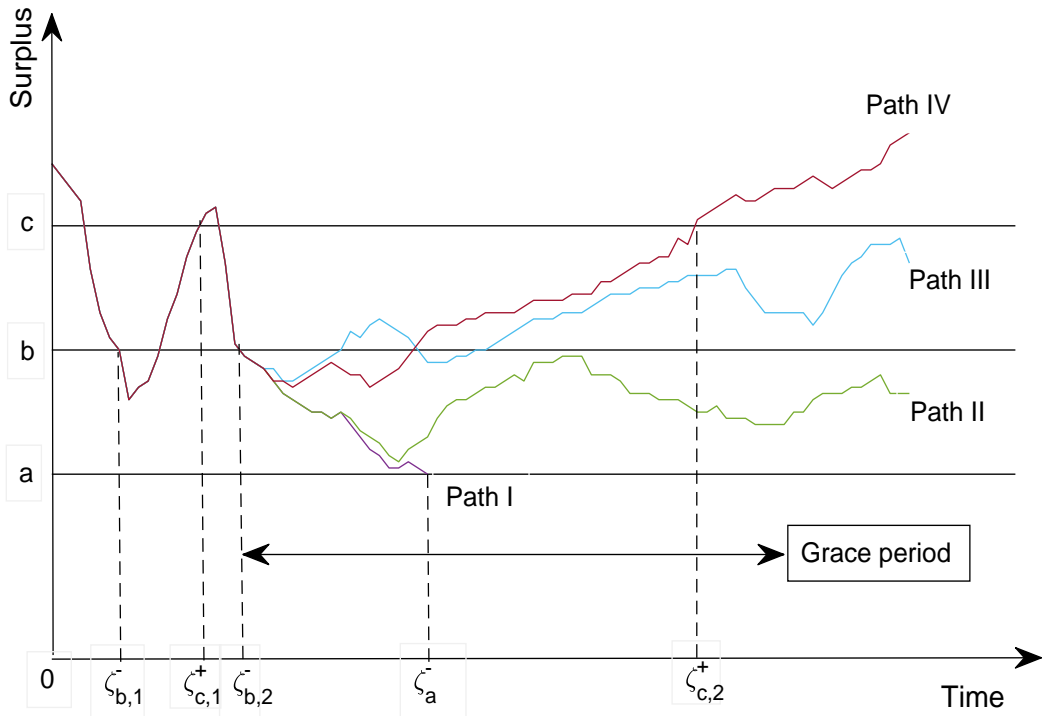}
\caption{The bankruptcy scenarios.}
\end{figure}

Figure 1 gives four possible paths of an insurer's surplus process. In the setting of spectrally negative L\'{e}vy processes, jumps are expected to occur. To avoid bewildering array of nodes in the coordinate axes, we omit the jumps to focus more on the idea of {\it solvent} and {\it insolvent} states. As sketched in Figure 1, from time 0 to time $\zeta^{-}_{b,1}$, the surplus process stays above barrier $b$ and hence the state of {\it solvent} applies. As the financial situation gets worse, the insurer transferred to {\it insolvent} state from $\zeta^{-}_{b,1}$ until the surplus recovers to the safety barrier $c$ at time $\zeta^{+}_{c,1}$, where the state of {\it solvent} resumes. At $\zeta^{-}_{b,2}$, the insurer turns to be {\it insolvent} again. We depict four possible scenarios from $\zeta^{-}_{b,2}$ of the surplus process. Path I represents the case of direct liquidation during the insolvent period with the breaching of barrier $a$. Path II and Path III never touch the liquidation barrier $a$ but have been staying below the safety barrier $c$ for so long that is unacceptable from the regulator, hence the insurer will still be concluded to liquidation. Path IV climbs up to the safety barrier $c$ successfully within the grace period, hence from $\zeta^{+}_{c,2}$ onward, the insurer stays at the state of {\it solvent} in Path IV.

\begin{rem}
\label{barrier_state}
Note that the way we define states using three barriers is different from Broadie et al. (2007), where a three-barrier model was also adopted but defined four states: liquid, equity dilution, default and liquidation. They used the three barriers directly as the division boundaries. In this paper, we only define three states of the surplus process as the possible outputs of an insurer: solvent, insolvent and liquidation. From technical point of view, four states of the surplus process is also tractable in mathematics but may lead to tedious expressions. The track of defining states in this paper is consistent with Li et al. (2020), which is practically workable in the context of insurance industry, and also produces neat mathematical results.

\end{rem}

\begin{rem}
\label{special_liquidation}
As pointed out by Li et al. (2020), the definition of liquidation time in \eqref{liquidation} includes
several existing stopping times of insolvency as special cases. For example, as the liquidation barrier $a$ goes to $-\infty$ and let the rehabilitation barrier $b$ equal to the safety barrier $c$, then the liquidation time $T$ retrieves the so-called Parisian ruin time, which is a popular topic of interest in recent years. We refer to Remark 2.1 in Li et al. (2020) for more examples and their applications.

\end{rem}

%
\section{Main Results}
\setcounter{section}{3} \setcounter{equation}{0}

This section aims to solve the Gerber-Shiu function at the time of liquidation in terms of the scale functions and the L\'evy triplet associated with $X$ and $\widetilde{X}$. As consequences, expressions of the discounted joint  probability density function of the liquidation time, the surplus at liquidation and the historical high of the surplus until liquidation, the Laplace transform of the liquidation time conditional on that liquidation occurs prior to the first up-crossing time of some fixed level, as well as the probability of liquidation are also derived.


By \eqref{mar.} and \eqref{mar.1}, we define for $q\geq 0$, $\lambda>0$, $a<b<c$, $b\leq x\leq z$ and $a\leq w\leq c$, an auxiliary function as follows
\begin{eqnarray}\label{om.w.l}
\hspace{-0.2cm}\Omega^{(q,q+\lambda)}_{\phi}(w,b,x,z)
\hspace{-0.3cm}&:=&\hspace{-0.3cm}\mathbb{E}_{x}\left[\mathrm{e}^{-q\tau_{b}^{-}}\phi_{q+\lambda}(X_{\tau_{b}^{-}}-w)\mathbf{1}_
{\{\tau_{b}^{-}<\tau_{z}^{+}\}}\right]\nonumber\\
\hspace{-0.3cm}&&\hspace{-3cm}=\frac{\sigma^2}{2}\phi_{q+\lambda}(b-w)\bigg[W_{q}^{\prime}(x-b)
-W_{q}(x-b)\frac{W_{q}^{\prime}(z-b)}{W_{q}(z-b)}\bigg]\nonumber\\
\hspace{-0.5cm}&&\hspace{-3cm}\quad\,\,+\int_{0}^{z-b}\mathrm{d}y\int_{y}^{\infty}
\phi_{q+\lambda}(y-\theta+b-w)\upsilon(\mathrm{d}\theta)\bigg[\frac{W_{q}(z-b-y)}{W_{q}(z-b)}W_{q}(x-b)
-W_{q}(x-b-y)\bigg],
\end{eqnarray}
hence by \eqref{s.f.l.} one has
\begin{eqnarray}\label{om.w.}
\Omega^{(q,q+\lambda)}_{\phi}(w,b,x)
\hspace{-0.3cm}&:=&\hspace{-0.3cm}\lim_{z\rightarrow \infty}\Omega^{(q,q+\lambda)}_{\phi}(w,b,x,z)\nonumber\\
\hspace{-0.3cm}&&\hspace{-2.5cm}=\frac{\sigma^2}{2}\phi_{q+\lambda}(b-w)\big(W_{q}^{\prime}(x-b)
-\Phi_{q}W_{q}(x-b)\big)\nonumber\\
\hspace{-0.3cm}&&\hspace{-2.5cm}\quad\,\,+\int_{0}^{\infty}\mathrm{d}y\int_{y}^{\infty}
\phi_{q+\lambda}(y-\theta+b-w)\upsilon(\mathrm{d}\theta)\big(\mathrm{e}^{-\Phi_{q}y}W_{q}(x-b)
-W_{q}(x-b-y)\big).
\end{eqnarray}
The notations $\Omega^{(q,q+\lambda)}_{\phi}(w,b,x,z)$ and $\Omega^{(q,q+\lambda)}_{\phi}(w,b,x)$ will come up with $\phi=\mathbb{W},\mathbb{W}^{\prime},\mathbb{Z}$.
In particular, when $\mathbb{W}=W$ or $\mathbb{Z}=Z$, by Lemma 2.2 and (19) of Loeffen et al. (2014), one has the following concise expression via the scale functions
\begin{eqnarray}\label{lap.t.11.}
\Omega^{(q,q+\lambda)}_{\phi}(w,b,x,z)
\hspace{-0.3cm}&=&\hspace{-0.3cm}\phi_{q+\lambda}(x-w)-\lambda\int_{b}^{x}W_{q}(x-y)\phi_{q+\lambda}(y-w)\mathrm{d}y\nonumber\\
\hspace{-0.3cm}&&\hspace{-3cm}
-\frac{W_{q}(x-b)}{W_{q}(z-b)}\left(\phi_{q+\lambda}(z-w)-\lambda\int_{b}^{z}W_{q}(z-y)\phi_{q+\lambda}(y-w)\mathrm{d}y\right),\,\,\phi=W,\,Z.
\end{eqnarray}

The following Theorem \ref{theorem2} expresses the Gerber-Shiu function at the liquidation ruin time with exponentially distributed grace periods in terms of the scale functions and the L\'evy triplet associated with the L\'evy processes $X$ and $\widetilde{X}$.
\begin{thm} \label{theorem2}
For $q,\lambda\in[0,\infty)$, $a<b<c$ and a measurable function $f$, we have
\begin{eqnarray}\label{g.s.ff}
\hspace{-0.3cm}&&\hspace{-0.3cm}
\mathbb{E}_{x}\left[\mathrm{e}^{-qT}f\big(U_{T}\big)
\mathbf{1}_{\{T<\zeta_{z}^{+}\}}\right]
\nonumber\\
\hspace{-0.3cm}&=&\hspace{-0.3cm}\frac{\widetilde{\sigma}^2}{2}f(a)\left(\Omega^{(q,q+\lambda)}_{\mathbb{W}'}(a,b,x,z)-
\Omega^{(q,q+\lambda)}_{\mathbb{W}}(a,b,x,z)\frac{\mathbb{W}_{q+\lambda}^{\prime}(c\wedge z-a)}{\mathbb{W}_{q+\lambda}(c\wedge z-a)}\right)
\nonumber\\
\hspace{-0.3cm}&&\hspace{-0.3cm}+\int_{a}^{c\wedge z}\mathrm{d}y\int_{y-a}^{\infty}f(y-\theta)\widetilde{\upsilon}(\mathrm{d}\theta)\left(
\Omega^{(q,q+\lambda)}_{\mathbb{W}}(a,b,x,z)\frac{\mathbb{W}_{q+\lambda}(c\wedge z-y)}{\mathbb{W}_{q+\lambda}(c\wedge z-a)}
-\Omega^{(q,q+\lambda)}_{\mathbb{W}}(y,b,x,z)\right)
\nonumber\\
\hspace{-0.3cm}&&\hspace{-0.3cm}
+\lambda\int_{a}^{c\wedge z}f(y)\left(\Omega^{(q,q+\lambda)}_{\mathbb{W}}(a,b,x,z)\frac{\mathbb{W}_{q+\lambda}(c\wedge z-y)}
{\mathbb{W}_{q+\lambda}(c\wedge z-a)}-\Omega^{(q,q+\lambda)}_{\mathbb{W}}(y,b,x,z)\right)\mathrm{d}y
\nonumber\\
\hspace{-0.3cm}&&\hspace{-0.3cm}
+\int_{b}^{z}\mathrm{d}y\int_{y-a}^{\infty}
f(y-\theta)\upsilon(\mathrm{d}\theta)\bigg(\frac{W_{q}(z-y)}{W_{q}(z-b)}W_{q}(x-b)
-W_{q}(x-y)\bigg)
\nonumber\\
\hspace{-0.3cm}&&\hspace{-0.3cm}
+\frac{\Omega^{(q,q+\lambda)}_{\mathbb{W}}(a,b,x,z)
\,\mathbf{1}_{\{z> c\}}}{\mathbb{W}_{q+\lambda}(c-a)
-\Omega^{(q,q+\lambda)}_{\mathbb{W}}(a,b,c,z)}\left[\frac{\widetilde{\sigma}^2}{2}f(a)\left(\Omega^{(q,q+\lambda)}_{\mathbb{W}'}(a,b,c,z)-
\Omega^{(q,q+\lambda)}_{\mathbb{W}}(a,b,c,z)\frac{\mathbb{W}_{q+\lambda}^{\prime}(c-a)}{\mathbb{W}_{q+\lambda}(c-a)}\right)\right.\nonumber\\
\hspace{-0.3cm}&&\hspace{-0.3cm}\left.+\int_{a}^{c}\mathrm{d}y\int_{y-a}^{\infty}f(y-\theta)\widetilde{\upsilon}(\mathrm{d}\theta)\left(
\Omega^{(q,q+\lambda)}_{\mathbb{W}}(a,b,c,z)\frac{\mathbb{W}_{q+\lambda}(c-y)}{\mathbb{W}_{q+\lambda}(c-a)}
-\Omega^{(q,q+\lambda)}_{\mathbb{W}}(y,b,c,z)\right)\right.
\nonumber\\
\hspace{-0.3cm}&&\hspace{-0.3cm}+\lambda\int_{a}^{c}f(y)\left(\Omega^{(q,q+\lambda)}_{\mathbb{W}}(a,b,c,z)\frac{\mathbb{W}_{q+\lambda}(c-y)}
{\mathbb{W}_{q+\lambda}(c-a)}-\Omega^{(q,q+\lambda)}_{\mathbb{W}}(y,b,c,z)\right)\mathrm{d}y\nonumber\\
\hspace{-0.3cm}&&\hspace{-0.3cm}
\left.+\int_{b}^{z}\mathrm{d}y\int_{y-a}^{\infty}
f(y-\theta)\upsilon(\mathrm{d}\theta)\bigg(\frac{W_{q}(z-y)}{W_{q}(z-b)}W_{q}(c-b)
-W_{q}(c-y)\bigg)
\right], \quad b<x\leq z.
\end{eqnarray}
\end{thm}
\begin{proof} Denote by $\hbar(x)$ the left hand side of \eqref{g.s.ff}.
Note that, if liquidation occurs, then the surplus process $U$ has to down-cross the level $b$ beforehand, otherwise liquidation would never happen. In addition, one should bear in mind that, at the moment $U$ down-crosses $b$, $U$ may or may not down-cross $a$, in the former case liquidation occurs.
Using the piecewise strong Markov property of the surplus
process $U$, one obtains
\begin{eqnarray}\label{g.s.f.}
\hbar(x)
\hspace{-0.3cm}&=&\hspace{-0.3cm}
\mathbb{E}_{x}\left[\mathrm{e}^{-q\tau^{-}_{b}}\mathbf{1}_{\{\tau^{-}_{b}<\tau_{z}^{+}\}}
\mathbf{1}_{\{\tau^{-}_{b}<\tau_{a}^{-}\}}
\widetilde{\mathbb{E}}_{X_{\tau^{-}_{b}}}\left(\mathrm{e}^
{-q(\zeta^{-}_{a}\wedge\zeta_{b,c,e_{\lambda}})}f(U_{\zeta^{-}_{a}\wedge\zeta_{b,c,e_{\lambda}}})
\mathbf{1}_{\{\zeta^{-}_{a}\wedge\zeta_{b,c,e_{\lambda}}<\zeta_{z}^{+}\}}\right)\right]
\nonumber\\
\hspace{-0.3cm}&&\hspace{-0.3cm}
+\mathbb{E}_{x}\left[\mathrm{e}^{-q\tau^{-}_{b}}\mathbf{1}_{\{\tau^{-}_{b}<\tau_{z}^{+}\}}
\mathbf{1}_{\{\tau^{-}_{b}=\tau_{a}^{-}\}}f(X_{\tau^{-}_{b}})
\right]
\nonumber\\
\hspace{-0.3cm}&=&\hspace{-0.3cm}
\mathbb{E}_{x}\left[\mathrm{e}^{-q\tau^{-}_{b}}\mathbf{1}_{\{\tau^{-}_{b}<\tau_{z}^{+}\}}
\mathbf{1}_{\{X(\tau^{-}_{b})\geq a\}}
\widetilde{\mathbb{E}}_{X_{\tau^{-}_{b}}}\left(\mathrm{e}^
{-q(\tau^{-}_{a}\wedge e_{\lambda})}f(U_{\tau^{-}_{a}\wedge e_{\lambda}})
\mathbf{1}_{\{\tau^{-}_{a}\wedge e_{\lambda}<\tau^{+}_{c}\}}\right)\right]
\nonumber\\
\hspace{-0.3cm}&&\hspace{-0.3cm}
+\mathbb{E}_{x}\left[\mathrm{e}^{-q\tau^{-}_{b}}\mathbf{1}_{\{\tau^{-}_{b}<\tau_{z}^{+}\}}
\mathbf{1}_{\{X(\tau^{-}_{b})\geq a\}}
\widetilde{\mathbb{E}}_{X_{\tau^{-}_{b}}}\left(\mathrm{e}^
{-q \tau^{+}_{c}}
\mathbf{1}_{\{\tau^{+}_{c}<\tau^{-}_{a}\wedge e_{\lambda}\}}\right)\hbar(c)\right]
\nonumber\\
\hspace{-0.3cm}&&\hspace{-0.3cm}
+\mathbb{E}_{x}\left[\mathrm{e}^{-q\tau^{-}_{b}}\mathbf{1}_{\{\tau^{-}_{b}<\tau_{z}^{+}\}}
\mathbf{1}_{\{X(\tau^{-}_{b})< a\}}f(X_{\tau^{-}_{b}})
\right],\quad
b<x\leq z, \, c\leq z.
\end{eqnarray}
It is seen that the inner expectation in the first term on the right hand side of \eqref{g.s.f.} can be rewritten as
\begin{eqnarray}\label{g.s.f.1}
\hspace{-0.3cm}&&\hspace{-0.3cm}
\widetilde{\mathbb{E}}_{X_{\tau^{-}_{b}}}\left(\mathrm{e}^
{-q(\tau^{-}_{a}\wedge e_{\lambda})}f(U_{\tau^{-}_{a}\wedge e_{\lambda}})
\mathbf{1}_{\{\tau^{-}_{a}\wedge e_{\lambda}<\tau^{+}_{c}\}}\right)
\nonumber\\
\hspace{-0.3cm}&=&\hspace{-0.3cm}
\widetilde{\mathbb{E}}_{X_{\tau^{-}_{b}}}\left(\mathrm{e}^
{-q\tau^{-}_{a}}f(\widetilde{X}_{\tau^{-}_{a}})
\mathbf{1}_{\{\tau^{-}_{a}< e_{\lambda}\wedge\tau^{+}_{c}\}}\right)
+
\widetilde{\mathbb{E}}_{X_{\tau^{-}_{b}}}\left(\mathrm{e}^
{-q e_{\lambda}}f(\widetilde{X}_{e_{\lambda}})
\mathbf{1}_{\{ e_{\lambda}<\tau^{-}_{a}\wedge\tau^{+}_{c}\}}\right)
\nonumber\\
\hspace{-0.3cm}&=&\hspace{-0.3cm}
\widetilde{\mathbb{E}}_{X_{\tau^{-}_{b}}}\left(\mathrm{e}^
{-(q+\lambda)\tau^{-}_{a}}f(\widetilde{X}_{\tau^{-}_{a}})
\mathbf{1}_{\{\tau^{-}_{a}< \tau^{+}_{c}\}}\right)
+
\widetilde{\mathbb{E}}_{X_{\tau^{-}_{b}}}\left(\mathrm{e}^
{-q e_{\lambda}}f(\widetilde{X}_{e_{\lambda}})
\mathbf{1}_{\{ e_{\lambda}<\tau^{-}_{a}\wedge\tau^{+}_{c}\}}\right),
\end{eqnarray}
where in the second equality we have used the fact that $e_{\lambda}$ is independent of the process $\widetilde{X}$. Plugging \eqref{g.s.f.1} into \eqref{g.s.f.} gives
\begin{eqnarray}\label{g.s.f.2}
\hbar(x)
\hspace{-0.3cm}&=&\hspace{-0.3cm}
\mathbb{E}_{x}\left[\mathrm{e}^{-q\tau^{-}_{b}}\mathbf{1}_{\{\tau^{-}_{b}<\tau_{z}^{+}\}}
\mathbf{1}_{\{X(\tau^{-}_{b})\geq a\}}
\widetilde{\mathbb{E}}_{X_{\tau^{-}_{b}}}\left(\mathrm{e}^
{-(q+\lambda)\tau^{-}_{a}}f(\widetilde{X}_{\tau^{-}_{a}})
\mathbf{1}_{\{\tau^{-}_{a}< \tau^{+}_{c}\}}\right)\right]
\nonumber\\
\hspace{-0.3cm}&&\hspace{-0.3cm}
+
\mathbb{E}_{x}\left[\mathrm{e}^{-q\tau^{-}_{b}}\mathbf{1}_{\{\tau^{-}_{b}<\tau_{z}^{+}\}}
\mathbf{1}_{\{X(\tau^{-}_{b})\geq a\}}
\widetilde{\mathbb{E}}_{X_{\tau^{-}_{b}}}\left(\mathrm{e}^
{-q e_{\lambda}}f(\widetilde{X}_{e_{\lambda}})
\mathbf{1}_{\{ e_{\lambda}<\tau^{-}_{a}\wedge\tau^{+}_{c}\}}\right)\right]
\nonumber\\
\hspace{-0.3cm}&&\hspace{-0.3cm}
+\mathbb{E}_{x}\left[\mathrm{e}^{-q\tau^{-}_{b}}\mathbf{1}_{\{\tau^{-}_{b}<\tau_{z}^{+}\}}
\mathbf{1}_{\{X(\tau^{-}_{b})\geq a\}}
\widetilde{\mathbb{E}}_{X_{\tau^{-}_{b}}}\left(\mathrm{e}^
{-(q+\lambda) \tau^{+}_{c}}
\mathbf{1}_{\{\tau^{+}_{c}<\tau^{-}_{a}\}}\right)\right]\hbar(c)
\nonumber\\
\hspace{-0.3cm}&&\hspace{-0.3cm}
+\mathbb{E}_{x}\left[\mathrm{e}^{-q\tau^{-}_{b}}\mathbf{1}_{\{\tau^{-}_{b}<\tau_{z}^{+}\}}
\mathbf{1}_{\{X(\tau^{-}_{b})< a\}}f(X_{\tau^{-}_{b}})
\right],\quad
b<x\leq z, \, c\leq z.
\end{eqnarray}
By \eqref{mar.}, \eqref{mar.1} and the fact that $\mathbb{W}_{q+\lambda}$ and $\mathbb{W}_{q+\lambda}^{\prime}$ vanishes on the negative half line, the first term on the right hand side of \eqref{g.s.f.2} is
\begin{eqnarray}\label{g.s.f1.}
\hspace{-0.3cm}&&\hspace{-0.3cm}
\mathbb{E}_{x}\Bigg[\mathrm{e}^{-q\tau^{-}_{b}}\mathbf{1}_{\{\tau^{-}_{b}<\tau^{+}_{z}\}}
\Bigg[\frac{{\widetilde{\sigma}}^2}{2}f(a)\left(\mathbb{W}_{q+\lambda}^{\prime}
(X_{\tau^{-}_{b}}-a)
-\mathbb{W}_{q+\lambda}(X_{\tau^{-}_{b}}-a)\frac{\mathbb{W}_{q+\lambda}^{\prime}(c-a)}{\mathbb{W}_{q+\lambda}(c-a)}\right)
\nonumber\\
\hspace{-0.3cm}&&\hspace{-0.3cm}+\int_{a}^{c}\mathrm{d}y\int_{y-a}^{\infty}
f(y-\theta)\widetilde{\upsilon}(\mathrm{d}\theta)\left(\frac{\mathbb{W}_{q+\lambda}(c-y)}{\mathbb{W}_{q+\lambda}(c-a)}
\mathbb{W}_{q+\lambda}(X_{\tau^{-}_{b}}-a)-\mathbb{W}_{q+\lambda}(X_{\tau^{-}_{b}}-y)\right)\Bigg]\Bigg]\nonumber\\
\hspace{-0.3cm}&=&\hspace{-0.3cm}
\frac{\widetilde{\sigma}^2}{2}f(a)\left(\Omega^{(q,q+\lambda)}_{\mathbb{W}'}(a,b,x,z)-\frac{\mathbb{W}_{q+\lambda}^{\prime}(c-a)}{\mathbb{W}_{q+\lambda}(c-a)}
\Omega^{(q,q+\lambda)}_{\mathbb{W}}(a,b,x,z)\right)\nonumber\\
\hspace{-0.3cm}&&\hspace{-0.3cm}+\int_{a}^{c}\mathrm{d}y\int_{y-a}^{\infty}
f(y-\theta)\widetilde{\upsilon}(\mathrm{d}\theta)\left(\frac{\mathbb{W}_{q+\lambda}(c-y)}{\mathbb{W}_{q+\lambda}(c-a)}
\Omega^{(q,q+\lambda)}_{\mathbb{W}}(a,b,x,z)-\Omega^{(q,q+\lambda)}_{\mathbb{W}}(y,b,x,z)\right).
\end{eqnarray}
By \eqref{r.m.}, the second term on the right hand side of \eqref{g.s.f.2} is
\begin{eqnarray}\label{g.s.f2.}
\hspace{-0.3cm}&&\hspace{-0.3cm}
\lambda\int_{a}^{c}f(y)\mathbb{E}_{x}\left[\mathrm{e}^{-q\tau^{-}_{b}}\left(\frac{\mathbb{W}_{q+\lambda}(X_{\tau^{-}_{b}}-a)
\mathbb{W}_{q+\lambda}(c-y)}{\mathbb{W}_{q+\lambda}(c-a)}-\mathbb{W}_{q+\lambda}(X_{\tau^{-}_{b}}-y)\right)
\mathbf{1}_{\{\tau^{-}_{b}<\tau^{+}_{z}\}}\right]\mathrm{d}y
\nonumber\\
\hspace{-0.3cm}&=&\hspace{-0.3cm}
\lambda\int_{a}^{c}f(y)\left(\frac{\mathbb{W}_{q+\lambda}(c-y)}{\mathbb{W}_{q+\lambda}(c-a)}\Omega^{(q,q+\lambda)}_{\mathbb{W}}(a,b,x,z)
-\Omega^{(q,q+\lambda)}_{\mathbb{W}}(y,b,x,z)\right)\mathrm{d}y.
\end{eqnarray}
By the space homogeneity of $\widetilde{X}$ and \eqref{two.side.e.}, one can write the third term on the right hand side of \eqref{g.s.f.2} as
\begin{eqnarray}\label{g.s.f3.}
\hspace{-0.3cm}&&\hspace{-0.3cm}
\mathbb{E}_{x}\left[\mathrm{e}^{-q\tau^{-}_{b}}\mathbf{1}_{\{\tau^{-}_{b}<\tau^{+}_{z}\}}
\mathbf{1}_{\{X(\tau^{-}_{b})\geq a\}}
\widetilde{\mathbb{E}}_{X_{\tau^{-}_{b}}-a}\left(\mathrm{e}^
{-(q+\lambda) \tau_{c-a}^{+}}\mathbf{1}_{\{\tau_{c-a}^{+}<\tau^{-}_{0}\}}\right)\right]\hbar(c)
\nonumber\\
\hspace{-0.3cm}&=&\hspace{-0.3cm}
\mathbb{E}_{x}\left[\mathrm{e}^{-q\tau^{-}_{b}}\frac{\mathbb{W}_{q+\lambda}(X_{\tau^{-}_{b}}-a)}{\mathbb{W}_{q+\lambda}(c-a)}\mathbf{1}_{\{\tau^{-}_{b}<\tau^{+}_{z}\}}\right]\hbar(c)
\nonumber\\
\hspace{-0.3cm}&=&\hspace{-0.3cm}
\frac{\Omega^{(q,q+\lambda)}_{\mathbb{W}}(a,b,x,z)}{\mathbb{W}_{q+\lambda}(c-a)}\hbar(c).
\end{eqnarray}
By \eqref{mar.} and \eqref{mar.1}, the fourth term on the right hand side of \eqref{g.s.f.2} is
\begin{eqnarray}
\label{g.s.f4.}
\hspace{-0.3cm}
\int_{0}^{z-b}\mathrm{d}y\int_{y+b-a}^{\infty}
f(y-\theta+b)\upsilon(\mathrm{d}\theta)\bigg[\frac{W_{q}(z-b-y)}{W_{q}(z-b)}W_{q}(x-b)
-W_{q}(x-b-y)\bigg]
,\quad
x\in(b,\infty).
\end{eqnarray}
Combining \eqref{g.s.f.2}, \eqref{g.s.f1.}, \eqref{g.s.f2.}, \eqref{g.s.f3.} and \eqref{g.s.f4.}
and then letting $x=c$ yields
\begin{eqnarray}\label{g.s.f.c.}
\hbar(c)
\hspace{-0.3cm}&=&\hspace{-0.3cm}
\frac{\mathbb{W}_{q+\lambda}(c-a)}{\mathbb{W}_{q+\lambda}(c-a)
-\Omega^{(q,q+\lambda)}_{\mathbb{W}}(a,b,c,z)}
\left[\frac{\widetilde{\sigma}^2}{2}f(a)\Bigg[\Omega^{(q,q+\lambda)}_{\mathbb{W}'}(a,b,c,z)-
\Omega^{(q,q+\lambda)}_{\mathbb{W}}(a,b,c,z)\frac{\mathbb{W}_{q+\lambda}^{\prime}(c-a)}{\mathbb{W}_{q+\lambda}(c-a)}\Bigg]\right.\nonumber\\
\hspace{-0.3cm}&&\hspace{-0.3cm}\left.+\int_{a}^{c}\mathrm{d}y\int_{y-a}^{\infty}f(y-\theta)\widetilde{\upsilon}(\mathrm{d}\theta)\left(
\Omega^{(q,q+\lambda)}_{\mathbb{W}}(a,b,c,z)\frac{\mathbb{W}_{q+\lambda}(c-y)}{\mathbb{W}_{q+\lambda}(c-a)}
-\Omega^{(q,q+\lambda)}_{\mathbb{W}}(y,b,c,z)\right)\right.\nonumber\\
\hspace{-0.3cm}&&\hspace{-0.3cm}\left.+\lambda\int_{a}^{c}f(y)\left(\Omega^{(q,q+\lambda)}_{\mathbb{W}}(a,b,c,z)\frac{\mathbb{W}_{q+\lambda}(c-y)}
{\mathbb{W}_{q+\lambda}(c-a)}-\Omega^{(q,q+\lambda)}_{\mathbb{W}}(y,b,c,z)\right)\mathrm{d}y \right.\nonumber\\
\hspace{-0.3cm}&&\hspace{-0.3cm}\left.+\int_{b}^{z}\mathrm{d}y\int_{y-a}^{\infty}
f(y-\theta)\upsilon(\mathrm{d}\theta)\bigg(\frac{W_{q}(z-y)}{W_{q}(z-b)}W_{q}(c-b)
-W_{q}(c-y)\bigg)\right],\nonumber
\end{eqnarray}
which together with \eqref{g.s.f.2}, \eqref{g.s.f1.}, \eqref{g.s.f2.} \eqref{g.s.f3.} and \eqref{g.s.f4.} leads to \eqref{g.s.ff}.

For the case of $b<x\leq z\leq c$, by similar arguments as used above for the case that $b<x\leq z$ and $c\leq z$, one has
\begin{eqnarray}\label{g.s.f.add.}
\hbar(x)
\hspace{-0.3cm}&=&\hspace{-0.3cm}
\mathbb{E}_{x}\left[\mathrm{e}^{-q\tau^{-}_{b}}\mathbf{1}_{\{\tau^{-}_{b}<\tau_{z}^{+}\}}
\mathbf{1}_{\{X(\tau^{-}_{b})> a\}}
\widetilde{\mathbb{E}}_{X_{\tau^{-}_{b}}}\left(\mathrm{e}^
{-(q+\lambda)\tau^{-}_{a}}f(\widetilde{X}_{\tau^{-}_{a}})
\mathbf{1}_{\{\tau^{-}_{a}< \tau^{+}_{z}\}}\right)\right]
\nonumber\\
\hspace{-0.3cm}&&\hspace{-0.3cm}
+
\mathbb{E}_{x}\left[\mathrm{e}^{-q\tau^{-}_{b}}\mathbf{1}_{\{\tau^{-}_{b}<\tau_{z}^{+}\}}
\mathbf{1}_{\{X(\tau^{-}_{b})> a\}}
\widetilde{\mathbb{E}}_{X_{\tau^{-}_{b}}}\left(\mathrm{e}^
{-q e_{\lambda}}f(\widetilde{X}_{e_{\lambda}})
\mathbf{1}_{\{ e_{\lambda}<\tau^{-}_{a}\wedge\tau^{+}_{z}\}}\right)\right]
\nonumber\\
\hspace{-0.3cm}&&\hspace{-0.3cm}
+\mathbb{E}_{x}\left[\mathrm{e}^{-q\tau^{-}_{b}}\mathbf{1}_{\{\tau^{-}_{b}<\tau_{z}^{+}\}}
\mathbf{1}_{\{X(\tau^{-}_{b})\leq a\}}f(X_{\tau^{-}_{b}})
\right],\quad
b<x\leq z\leq c,\nonumber
\end{eqnarray}
which together with \eqref{r.m.}, \eqref{mar.} and \eqref{mar.1} yields \eqref{g.s.ff}.
The proof of Theorem \ref{theorem2} is complete.
\end{proof}

\medskip\medskip
Define $\overline{U}_{t}:=\sup_{0\leq s\leq t}U_{s}$, then $T<\zeta_{z}^{+}$ is equivalent to $\overline{U}_{T}< z$. Given a sequence of exponentially distributed grace periods, the following Corollary \ref{cor.2.} characterizes the joint discounted probability density function of the liquidation time, the surplus at and the running supreme of the surplus until the liquidation time.
\begin{cor}\label{cor.2.}
For $q\geq 0$, $\lambda>0$, $a<b<c$ and $b<x\leq z$, we have
\begin{eqnarray}\label{d1}
\hspace{-0.3cm}&&\hspace{-0.3cm}
\mathbb{E}_{x}\left(\mathrm{e}^{-qT}, U_{T}\in\mathrm{d}u, \overline{U}_{T}\in\mathrm{d}z \right)
\nonumber\\
\hspace{-0.3cm}&=&\hspace{-0.3cm}
\int^{c}_{a}\left(\frac{\partial\Omega^{(q,q+\lambda)}_{\mathbb{W}}(a,b,x,z)}{\partial z}\frac{\mathbb{W}_{q+\lambda}(c-y)}{\mathbb{W}_{q+\lambda}(c-a)}-\frac{\partial\Omega^{(q,q+\lambda)}_{\mathbb{W}}(y,b,x,z)}{\partial z}\right)\widetilde{\upsilon}(y-\mathrm{d}u)\mathbf{1}_{\{z>c\}}\mathrm{d}y\,\mathrm{d}z
\nonumber\\
\hspace{-0.3cm}&&\hspace{-0.3cm}
+\int^{z}_{a}\frac{\partial}{\partial z}\left(\Omega^{(q,q+\lambda)}_{\mathbb{W}}(a,b,x,z)\frac{\mathbb{W}_{q+\lambda}(z-y)}{\mathbb{W}_{q+\lambda}(z-a)}-\Omega^{(q,q+\lambda)}_{\mathbb{W}}(y,b,x,z) \right)\widetilde{\upsilon}(y-\mathrm{d}u)\mathbf{1}_{\{z\leq c\}}\mathrm{d}y\,\mathrm{d}z\nonumber\\
\hspace{-0.3cm}&&\hspace{-0.3cm}
+\Omega^{(q,q+\lambda)}_{\mathbb{W}}(a,b,x,z)\frac{\mathbb{W}_{q+\lambda}(0+)}{\mathbb{W}_{q+\lambda}(z-a)}\widetilde{\upsilon}(z-\mathrm{d}u)
\mathrm{d}z\mathbf{1}_{\{z\leq c\}}
\nonumber\\
\hspace{-0.3cm}&&\hspace{-0.3cm}
+
\frac{W_{q}(0+)}{W_{q}(z-b)}W_{q}(x-b)\upsilon(z-\mathrm{d}u)\mathrm{d}z
+\int_{b}^{z}
\frac{\partial}{\partial z}\bigg[\frac{W_{q}(z-y)}{W_{q}(z-b)}\bigg]W_{q}(x-b)
\upsilon(y-\mathrm{d}u)\mathrm{d}y\mathrm{d}z
\nonumber\\
\hspace{-0.3cm}&&\hspace{-0.3cm}
+\frac{\Omega^{(q,q+\lambda)}_{\mathbb{W}}(a,b,x,z)\mathbf{1}_{\{z>c\}}}{\mathbb{W}_{q+\lambda}(c-a)
-\Omega^{(q,q+\lambda)}_{\mathbb{W}}(a,b,c,z)}\Bigg[\int^{c}_{a}\left(\frac{\partial\Omega^{(q,q+\lambda)}_{\mathbb{W}}(a,b,c,z)}{\partial z}\frac{\mathbb{W}_{q+\lambda}(c-y)}{\mathbb{W}_{q+\lambda}(c-a)}\right.
\nonumber\\
\hspace{-0.3cm}&&\hspace{-0.3cm}
\left.-\frac{\partial\Omega^{(q,q+\lambda)}_{\mathbb{W}}(y,b,c,z)}{\partial z}\right)\widetilde{\upsilon}(y-\mathrm{d}u)\mathrm{d}y\,\mathrm{d}z
+\int_{b}^{z}
\frac{\partial}{\partial z}\bigg[\frac{W_{q}(z-y)}{W_{q}(z-b)}\bigg]W_{q}(c-b)
\upsilon(y-\mathrm{d}u)\mathrm{d}y\mathrm{d}z
\nonumber\\
\hspace{-0.3cm}&&\hspace{-0.3cm}
+\mathrm{d}z\upsilon(z-\mathrm{d}u)
\frac{W_{q}(0+)}{W_{q}(z-b)}W_{q}(c-b)
\Bigg]
+
\frac{\partial}{\partial z}\Bigg[\frac{\Omega^{(q,q+\lambda)}_{\mathbb{W}}(a,b,x,z)}{\mathbb{W}_{q+\lambda}(c-a)
-\Omega^{(q,q+\lambda)}_{\mathbb{W}}(a,b,c,z)}\Bigg]
\nonumber\\
\hspace{-0.3cm}&&\hspace{-0.3cm}
\times\Bigg[\int^{c}_{a}\left(\Omega^{(q,q+\lambda)}_{\mathbb{W}}(a,b,c,z) \frac{\mathbb{W}_{q+\lambda}(c-y)}{\mathbb{W}_{q+\lambda}(c-a)}-\Omega^{(q,q+\lambda)}_{\mathbb{W}}(y,b,c,z)\right)\widetilde{\upsilon}(y-\mathrm{d}u)\mathrm{d}y\,\mathrm{d}z
\nonumber\\
\hspace{-0.3cm}&&\hspace{-0.3cm}
+\int_{b}^{z}
\bigg(\frac{W_{q}(z-y)}{W_{q}(z-b)}W_{q}(c-b)
-W_{q}(c-y)\bigg)\upsilon(y-\mathrm{d}u)\mathrm{d}y\mathrm{d}z
\Bigg]\mathbf{1}_{\{z>c\}},\nonumber
\end{eqnarray}
for $u\in(-\infty,a)$, and
\begin{eqnarray}\label{d2}
\hspace{-0.3cm}&&\hspace{-0.3cm}
\mathbb{E}_{x}\left[\mathrm{e}^{-qT}, U_{T}\in\mathrm{d}u, \overline{U}_{T}\in\mathrm{d}z \right]\nonumber\\
\hspace{-0.3cm}&=&\hspace{-0.3cm}
\lambda\,\mathrm{d}u\,\mathrm{d}z\bigg[\frac{\frac{\partial\Omega^{(q,q+\lambda)}_{\mathbb{W}}(a,b,x,z)}{\partial z}\mathbb{W}_{q+\lambda}(c-u)\mathbf{1}_{\{z>c\}}}
{\mathbb{W}_{q+\lambda}(c-a)}
+\frac{\partial\Omega^{(q,q+\lambda)}_{\mathbb{W}}(a,b,x,z)\frac{\mathbb{W}_{q+\lambda}(z-u)}
{\mathbb{W}_{q+\lambda}(z-a)}}{\partial z}\mathbf{1}_{\{z\leq c\}}
\nonumber\\
\hspace{-0.3cm}&&\hspace{-0.3cm}
+\frac{\Omega^{(q,q+\lambda)}_{\mathbb{W}}(a,b,x,z)\mathbf{1}_{\{z>c\}}}{\mathbb{W}_{q+\lambda}(c-a)
-\Omega^{(q,q+\lambda)}_{\mathbb{W}}(a,b,c,z)}
\bigg[\frac{\partial\Omega^{(q,q+\lambda)}_{\mathbb{W}}(a,b,c,z)}{\partial z}\frac{\mathbb{W}_{q+\lambda}(c-u)}
{\mathbb{W}_{q+\lambda}(c-a)}-\frac{\partial\Omega^{(q,q+\lambda)}_{\mathbb{W}}(u,b,c,z)}{\partial z}\bigg]
\nonumber\\
\hspace{-0.3cm}&&\hspace{-0.3cm}
+\frac{\partial}{\partial z}\bigg[\frac{\Omega^{(q,q+\lambda)}_{\mathbb{W}}(a,b,x,z)}{\mathbb{W}_{q+\lambda}(c-a)
-\Omega^{(q,q+\lambda)}_{\mathbb{W}}(a,b,c,z)}\bigg]\bigg[\Omega^{(q,q+\lambda)}_{\mathbb{W}}(a,b,c,z)\frac{\mathbb{W}_{q+\lambda}(c-u)}
{\mathbb{W}_{q+\lambda}(c-a)}-\Omega^{(q,q+\lambda)}_{\mathbb{W}}(u,b,c,z)
\bigg]\mathbf{1}_{\{z>c\}}\bigg]
\nonumber\\
\hspace{-0.3cm}&&\hspace{-0.3cm}
-\lambda\,\mathrm{d}u\,\mathrm{d}z\frac{\partial\Omega^{(q,q+\lambda)}_{\mathbb{W}}(u,b,x,z)}{\partial z}
-\frac{\widetilde{\sigma}^2}{2}\delta_{a}(\mathrm{d}u)\,\mathrm{d}z\frac{\partial}{\partial z}\left(\Omega^{(q,q+\lambda)}_{\mathbb{W}}(a,b,x,z)\frac{\mathbb{W}_{q+\lambda}^{\prime}(z-a)}{\mathbb{W}_{q+\lambda}(z-a)}\right)\mathbf{1}_{\{z\leq c\}}
\nonumber\\
\hspace{-0.3cm}&&\hspace{-0.3cm}
+
\frac{\widetilde{\sigma}^2}{2}\delta_{a}(\mathrm{d}u)\,\mathrm{d}z\Bigg[\frac{\partial\Omega^{(q,q+\lambda)}_{\mathbb{W}'}(a,b,x,z)}{\partial z}-
\frac{\partial\Omega^{(q,q+\lambda)}_{\mathbb{W}}(a,b,x,z)}{\partial z}\frac{\mathbb{W}_{q+\lambda}^{\prime}(c-a)}{\mathbb{W}_{q+\lambda}(c-a)}\mathbf{1}_{\{z>c\}}
\nonumber\\
\hspace{-0.3cm}&&\hspace{-0.3cm}
+\frac{\Omega^{(q,q+\lambda)}_{\mathbb{W}}(a,b,x,z)\mathbf{1}_{\{z>c\}}}{\mathbb{W}_{q+\lambda}(c-a)
-\Omega^{(q,q+\lambda)}_{\mathbb{W}}(a,b,c,z)}\bigg[\frac{\partial\Omega^{(q,q+\lambda)}_{\mathbb{W}'}(a,b,c,z)}{\partial z}-
\frac{\partial\Omega^{(q,q+\lambda)}_{\mathbb{W}}(a,b,c,z)}{\partial z}\frac{\mathbb{W}_{q+\lambda}^{\prime}(c-a)}{\mathbb{W}_{q+\lambda}(c-a)}\bigg]
\nonumber\\
\hspace{-0.3cm}&&\hspace{-0.3cm}
+\frac{\partial}{\partial z}\bigg[\frac{\Omega^{(q,q+\lambda)}_{\mathbb{W}}(a,b,x,z)}{\mathbb{W}_{q+\lambda}(c-a)
-\Omega^{(q,q+\lambda)}_{\mathbb{W}}(a,b,c,z)}\bigg]
\bigg[\Omega^{(q,q+\lambda)}_{\mathbb{W}'}(a,b,c,z)-
\Omega^{(q,q+\lambda)}_{\mathbb{W}}(a,b,c,z)\frac{\mathbb{W}_{q+\lambda}^{\prime}(c-a)}{\mathbb{W}_{q+\lambda}(c-a)}\bigg]\mathbf{1}_{\{z>c\}}\bigg]
,\nonumber
\end{eqnarray}
for $u\in[a,c\wedge z]$.
Here $\delta_{a}(\mathrm{d}u)$ denotes the Dirac measure which assigns unit mass to the singleton $\{a\}$.
\end{cor}

\begin{proof}
Let $f(w) = \mathbf{1}_{\{w\leq u\}}$ in (\ref{g.s.ff}) one may have
\begin{eqnarray}\label{11}
\hspace{-0.3cm}&&\hspace{-0.3cm}
\mathbb{E}_{x}\left[\mathrm{e}^{-qT}\mathbf{1}_{\{U_{T}\leq u\}}\mathbf{1}_{\{\overline{U}_{T}< z\}}\right]
\nonumber\\
\hspace{-0.3cm}&=&\hspace{-0.3cm}\frac{\widetilde{\sigma}^2}{2}\mathbf{1}_{\{a\leq u\}}\left(\Omega^{(q,q+\lambda)}_{\mathbb{W}'}(a,b,x,z)-
\Omega^{(q,q+\lambda)}_{\mathbb{W}}(a,b,x,z)\frac{\mathbb{W}_{q+\lambda}^{\prime}(c\wedge z-a)}{\mathbb{W}_{q+\lambda}(c\wedge z-a)}\right)\nonumber\\
\hspace{-0.3cm}&&\hspace{-0.3cm}+\int_{a}^{c\wedge z}\mathrm{d}y\int_{y-a}^{\infty}\mathbf{1}_{\{y-\theta\leq u\}}\widetilde{\upsilon}(\mathrm{d}\theta)\left(
\Omega^{(q,q+\lambda)}_{\mathbb{W}}(a,b,x,z)\frac{\mathbb{W}_{q+\lambda}(c\wedge z-y)}{\mathbb{W}_{q+\lambda}(c\wedge z-a)}
-\Omega^{(q,q+\lambda)}_{\mathbb{W}}(y,b,x,z)\right)\nonumber\\
\hspace{-0.3cm}&&\hspace{-0.3cm}+\lambda\int_{a}^{c\wedge z}\mathbf{1}_{\{y\leq u\}}\left(\Omega^{(q,q+\lambda)}_{\mathbb{W}}(a,b,x,z)\frac{\mathbb{W}_{q+\lambda}(c\wedge z-y)}
{\mathbb{W}_{q+\lambda}(c\wedge z-a)}-\Omega^{(q,q+\lambda)}_{\mathbb{W}}(y,b,x,z)\right)\mathrm{d}y
\nonumber\\
\hspace{-0.3cm}&&\hspace{-0.3cm}
+\int_{b}^{z}\mathrm{d}y\int_{y-a}^{\infty}
\mathbf{1}_{\{y-\theta\leq u\}}\upsilon(\mathrm{d}\theta)\bigg(\frac{W_{q}(z-y)}{W_{q}(z-b)}W_{q}(x-b)
-W_{q}(x-y)\bigg)
\nonumber\\
\hspace{-0.3cm}&&\hspace{-0.3cm}+\frac{\Omega^{(q,q+\lambda)}_{\mathbb{W}}(a,b,x,z)\mathbf{1}_{\{z>c\}}}{\mathbb{W}_{q+\lambda}(c-a)
-\Omega^{(q,q+\lambda)}_{\mathbb{W}}(a,b,c,z)}\left[\frac{\widetilde{\sigma}^2}{2}\Bigg[\Omega^{(q,q+\lambda)}_{\mathbb{W}'}(a,b,c,z)-
\Omega^{(q,q+\lambda)}_{\mathbb{W}}(a,b,c,z)\frac{\mathbb{W}_{q+\lambda}^{\prime}(c-a)}{\mathbb{W}_{q+\lambda}(c-a)}\Bigg]\mathbf{1}_{\{a\leq u\}}\right.\nonumber\\
\hspace{-0.3cm}&&\hspace{-0.3cm}\left.+\int_{a}^{c}\mathrm{d}y\int_{y-a}^{\infty}\mathbf{1}_{\{y-\theta\leq u\}}\widetilde{\upsilon}(\mathrm{d}\theta)\left(
\Omega^{(q,q+\lambda)}_{\mathbb{W}}(a,b,c,z)\frac{\mathbb{W}_{q+\lambda}(c-y)}{\mathbb{W}_{q+\lambda}(c-a)}
-\Omega^{(q,q+\lambda)}_{\mathbb{W}}(y,b,c,z)\right)\right.\nonumber\\
\hspace{-0.3cm}&&\hspace{-0.3cm}+\lambda\int_{a}^{c}\mathbf{1}_{\{y\leq u\}}\left(\Omega^{(q,q+\lambda)}_{\mathbb{W}}(a,b,c,z)\frac{\mathbb{W}_{q+\lambda}(c-y)}
{\mathbb{W}_{q+\lambda}(c-a)}-\Omega^{(q,q+\lambda)}_{\mathbb{W}}(y,b,c,z)\right)\mathrm{d}y
\nonumber\\
\hspace{-0.3cm}&&\hspace{-0.3cm}
+\int_{b}^{z}\mathrm{d}y\int_{y-a}^{\infty}
\mathbf{1}_{\{y-\theta\leq u\}}\upsilon(\mathrm{d}\theta)\bigg(\frac{W_{q}(z-y)}{W_{q}(z-b)}W_{q}(c-b)
-W_{q}(c-y)\bigg)
\Bigg], \quad b<x\leq z.\nonumber
\end{eqnarray}
Hence one has for $u\in(-\infty,a)$
\begin{eqnarray}\label{12}
\hspace{-0.3cm}&&\hspace{-0.3cm}
\mathbb{E}_{x}\left[\mathrm{e}^{-qT}\mathbf{1}_{\{U_{T}\leq u\}}\mathbf{1}_{\{\overline{U}_{T}< z\}}\right]
=
\frac{\Omega^{(q,q+\lambda)}_{\mathbb{W}}(a,b,x,z)\mathbf{1}_{\{z>c\}}}{\mathbb{W}_{q+\lambda}(c-a)
-\Omega^{(q,q+\lambda)}_{\mathbb{W}}(a,b,c,z)}
\nonumber\\
\hspace{-0.3cm}&&\hspace{0.6cm}
\times\left[\int_{a}^{c}\mathrm{d}y\int_{y-u}^{\infty}\widetilde{\upsilon}(\mathrm{d}\theta)\left(
\Omega^{(q,q+\lambda)}_{\mathbb{W}}(a,b,c,z)\frac{\mathbb{W}_{q+\lambda}(c-y)}{\mathbb{W}_{q+\lambda}(c-a)}
-\Omega^{(q,q+\lambda)}_{\mathbb{W}}(y,b,c,z)\right)\right.
\nonumber\\
\hspace{-0.3cm}&&\hspace{0.6cm}
\left.\quad\,\,+\int_{b}^{z}\mathrm{d}y\int_{y-u}^{\infty}
\upsilon(\mathrm{d}\theta)\bigg(\frac{W_{q}(z-y)}{W_{q}(z-b)}W_{q}(c-b)
-W_{q}(c-y)\bigg)\right]
\nonumber\\
\hspace{-0.3cm}&&\hspace{0.6cm}+
\int_{a}^{c\wedge z}\mathrm{d}y\int_{y-u}^{\infty}\widetilde{\upsilon}(\mathrm{d}\theta)\left(
\Omega^{(q,q+\lambda)}_{\mathbb{W}}(a,b,x,z)\frac{\mathbb{W}_{q+\lambda}(c\wedge z-y)}{\mathbb{W}_{q+\lambda}(c\wedge z-a)}
-\Omega^{(q,q+\lambda)}_{\mathbb{W}}(y,b,x,z)\right)
\nonumber\\
\hspace{-0.3cm}&&\hspace{0.6cm}
+\int_{b}^{z}\mathrm{d}y\int_{y-u}^{\infty}
\upsilon(\mathrm{d}\theta)\bigg(\frac{W_{q}(z-y)}{W_{q}(z-b)}W_{q}(x-b)
-W_{q}(x-y)\bigg), \quad b<x\leq z.
\end{eqnarray}
Meanwhile, for $u\in[a,c]$ one has
\begin{eqnarray}\label{13}
\hspace{-0.3cm}&&\hspace{-0.3cm}
\mathbb{E}_{x}\left[\mathrm{e}^{-qT}\mathbf{1}_{\{U_{T}\leq u\}}\mathbf{1}_{\{\overline{U}_{T}< z\}}\right]
\nonumber\\
\hspace{-0.3cm}&=&\hspace{-0.3cm}\frac{\widetilde{\sigma}^2}{2}\left(\Omega^{(q,q+\lambda)}_{\mathbb{W}'}(a,b,x,z)-
\Omega^{(q,q+\lambda)}_{\mathbb{W}}(a,b,x,z)\frac{\mathbb{W}_{q+\lambda}^{\prime}(c\wedge z-a)}{\mathbb{W}_{q+\lambda}(c\wedge z-a)}\right)\nonumber\\
\hspace{-0.3cm}&&\hspace{-0.3cm}+\int_{a}^{c\wedge z}\mathrm{d}y\int_{y-a}^{\infty}\widetilde{\upsilon}(\mathrm{d}\theta)\left(
\Omega^{(q,q+\lambda)}_{\mathbb{W}}(a,b,x,z)\frac{\mathbb{W}_{q+\lambda}(c\wedge z-y)}{\mathbb{W}_{q+\lambda}(c\wedge z-a)}
-\Omega^{(q,q+\lambda)}_{\mathbb{W}}(y,b,x,z)\right)\nonumber\\
\hspace{-0.3cm}&&\hspace{-0.3cm}+\lambda\int_{a}^{u}\left(\Omega^{(q,q+\lambda)}_{\mathbb{W}}(a,b,x,z)\frac{\mathbb{W}_{q+\lambda}(c\wedge z-y)}
{\mathbb{W}_{q+\lambda}(c\wedge z-a)}-\Omega^{(q,q+\lambda)}_{\mathbb{W}}(y,b,x,z)\right)\mathbf{1}_{\{z>u\}}\mathrm{d}y\nonumber\\
\hspace{-0.3cm}&&\hspace{-0.3cm}+\lambda\int_{a}^{z}\left(\Omega^{(q,q+\lambda)}_{\mathbb{W}}(a,b,x,z)\frac{\mathbb{W}_{q+\lambda}(z-y)}
{\mathbb{W}_{q+\lambda}(z-a)}-\Omega^{(q,q+\lambda)}_{\mathbb{W}}(y,b,x,z)\right)\mathbf{1}_{\{z\leq u\}}\mathrm{d}y
\nonumber\\
\hspace{-0.3cm}&&\hspace{-0.3cm}
+\int_{b}^{z}\mathrm{d}y\int_{y-a}^{\infty}
\upsilon(\mathrm{d}\theta)\bigg(\frac{W_{q}(z-y)}{W_{q}(z-b)}W_{q}(x-b)
-W_{q}(x-y)\bigg)
\nonumber\\
\hspace{-0.3cm}&&\hspace{-0.3cm}
+\frac{\Omega^{(q,q+\lambda)}_{\mathbb{W}}(a,b,x,z)\mathbf{1}_{\{z>c\}}}{\mathbb{W}_{q+\lambda}(c-a)
-\Omega^{(q,q+\lambda)}_{\mathbb{W}}(a,b,c,z)}\left[\frac{\widetilde{\sigma}^2}{2}\Bigg[\Omega^{(q,q+\lambda)}_{\mathbb{W}'}(a,b,c,z)-
\Omega^{(q,q+\lambda)}_{\mathbb{W}}(a,b,c,z)\frac{\mathbb{W}_{q+\lambda}^{\prime}(c-a)}{\mathbb{W}_{q+\lambda}(c-a)}\Bigg]\right.\nonumber\\
\hspace{-0.3cm}&&\hspace{-0.3cm}\left.+\int_{a}^{c}\mathrm{d}y\int_{y-a}^{\infty}\widetilde{\upsilon}(\mathrm{d}\theta)\left(
\Omega^{(q,q+\lambda)}_{\mathbb{W}}(a,b,c,z)\frac{\mathbb{W}_{q+\lambda}(c-y)}{\mathbb{W}_{q+\lambda}(c-a)}
-\Omega^{(q,q+\lambda)}_{\mathbb{W}}(y,b,c,z)\right)\right.\nonumber\\
\hspace{-0.3cm}&&\hspace{-0.3cm}+\lambda\int_{a}^{u}\left(\Omega^{(q,q+\lambda)}_{\mathbb{W}}(a,b,c,z)\frac{\mathbb{W}_{q+\lambda}(c-y)}
{\mathbb{W}_{q+\lambda}(c-a)}-\Omega^{(q,q+\lambda)}_{\mathbb{W}}(y,b,c,z)\right)\mathrm{d}y
\nonumber\\
\hspace{-0.3cm}&&\hspace{-0.3cm}
\left.+\int_{b}^{z}\mathrm{d}y\int_{y-a}^{\infty}
\upsilon(\mathrm{d}\theta)\bigg(\frac{W_{q}(z-y)}{W_{q}(z-b)}W_{q}(c-b)
-W_{q}(c-y)\bigg)\right], \quad b<x\leq z.
\end{eqnarray}
Now, differentiating both sides of (\ref{12}) and (\ref{13}) first in $u$
and then in $z$ yields the desired results. The proof is complete.
\end{proof}

\medskip\medskip
The following Corollary \ref{theorem1} gives the Laplace transform of the liquidation time with an exponential rehabilitation delay given that liquidation occurs prior to its first up-crossing time of the level $z$.
\begin{cor} \label{theorem1}
For $q,\lambda\in(0,\infty)$, $a<b<c<z$ and $b<x\leq z$, we have
\begin{eqnarray}\label{lap.tran.}
\mathbb{E}_{x}\bigg[\mathrm{e}^{-qT}\mathbf{1}_{\{T<\zeta^{+}_{z}\}}\bigg]
\hspace{-0.3cm}&=&\hspace{-0.3cm}
\frac{q}{q+\lambda}\left(\Omega^{(q,q+\lambda)}_{\mathbb{Z}}(a,b,x,z)-\frac{\mathbb{Z}_{q+\lambda}(c-a)}
{\mathbb{W}_{q+\lambda}(c-a)}\Omega^{(q,q+\lambda)}_{\mathbb{W}}(a,b,x,z)\right)\nonumber\\
\hspace{-0.5cm}&&\hspace{-3cm}
+\frac{\lambda}{q+\lambda}\left(Z_{q}(x-b)-\frac{Z_{q}(z-b)}{W_{q}(z-b)}W_{q}
(x-b)-\frac{\Omega^{(q,q+\lambda)}_{\mathbb{W}}(a,b,x,z)}{\mathbb{W}_{q+\lambda}(c-a)}\right)
\nonumber\\
\hspace{-0.5cm}&&\hspace{-3cm}+\frac{\Omega^{(q,q+\lambda)}_{\mathbb{W}}(a,b,x,z)}{\mathbb{W}_{q+\lambda}(c-a)
-\Omega^{(q,q+\lambda)}_{\mathbb{W}}(a,b,c,z)}\left[\frac{q}{q+\lambda}
\left(\Omega^{(q,q+\lambda)}_{\mathbb{Z}}(a,b,c,z)-\frac{\mathbb{Z}_{q+\lambda}(c-a)}
{\mathbb{W}_{q+\lambda}(c-a)}\Omega^{(q,q+\lambda)}_{\mathbb{W}}(a,b,c,z)\right)\right.
\nonumber\\
\hspace{-0.5cm}&&\hspace{-3cm}\left.+\frac{\lambda}{q+\lambda}\left(Z_{q}(c-b)-\frac{Z_{q}(z-b)}{W_{q}(z-b)}W_{q}
(c-b)-\frac{\Omega^{(q,q+\lambda)}_{\mathbb{W}}(a,b,c,z)}{\mathbb{W}_{q+\lambda}(c-a)}\right)
\right].
\end{eqnarray}
\end{cor}

\begin{proof}
By \eqref{two.side.d.}, \eqref{mar.}, \eqref{mar.1} and the extended definition that $\mathbb{Z}_{q+\lambda}=1$ and $\mathbb{W}_{q+\lambda}=0$ over the negative half line $(-\infty, 0)$, one has
\begin{eqnarray}\label{mar.0}
\hspace{-0.2cm}
\Bigg[\mathbb{Z}_{q+\lambda}(x-a)-\mathbb{Z}_{q+\lambda}(c-a)\frac{\mathbb{W}_{q+\lambda}(x-a)}{\mathbb{W}_{q+\lambda}(c-a)}\Bigg]\mathbf{1}_{\{x\leq b\}}
\hspace{-0.3cm}&=&\hspace{-0.3cm}
\Bigg[\frac{\widetilde{\sigma}^2}{2}\bigg[\mathbb{W}_{q+\lambda}^{\prime}(x-a)
-\mathbb{W}_{q+\lambda}(x-a)\frac{\mathbb{W}_{q+\lambda}^{\prime}(c-a)}{\mathbb{W}_{q+\lambda}(c-a)}\bigg]\nonumber\\
\hspace{-0.3cm}&&\hspace{-7.5cm}+\int_{0}^{c-a}\mathrm{d}y\int_{y}^{\infty}
\widetilde{\upsilon}(\mathrm{d}\theta)\bigg[\frac{\mathbb{W}_{q+\lambda}(c-a-y)}{\mathbb{W}_{q+\lambda}(c-a)}
\mathbb{W}_{q+\lambda}(x-a)-\mathbb{W}_{q+\lambda}(x-a-y)\bigg]\Bigg]\mathbf{1}_{\{a\leq x\leq b\}}
+\mathbf{1}_{\{x< a\}}.\nonumber
\end{eqnarray}
Replacing $x$ with $X_{\tau_{b}^{-}}$ in the above equation, multiplying the resulting equation with $\mathrm{e}^{-q\tau_{b}^{-}}\mathbf{1}_
{\{\tau_{b}^{-}<\tau_{z}^{+}\}}$, taking expectation under the probability measure $\mathbb{P}_{x}$ and then making use of \eqref{mar.} and \eqref{mar.1} once again, we have
\begin{eqnarray}\label{crt.1}
\hspace{-0.3cm}&&\hspace{-0.3cm}
\Omega^{(q,q+\lambda)}_{\mathbb{Z}}(a,b,x,z)-\frac{\mathbb{Z}_{q+\lambda}(c-a)}
{\mathbb{W}_{q+\lambda}(c-a)}\Omega^{(q,q+\lambda)}_{\mathbb{W}}(a,b,x,z)
\nonumber\\
\hspace{-0.3cm}&=&\hspace{-0.3cm}
\frac{\widetilde{\sigma}^2}{2}\left(\Omega^{(q,q+\lambda)}_{\mathbb{W}'}(a,b,x,z)-
\Omega^{(q,q+\lambda)}_{\mathbb{W}}(a,b,x,z)\frac{\mathbb{W}_{q+\lambda}^{\prime}(c-a)}{\mathbb{W}_{q+\lambda}(c-a)}\right)
\nonumber\\
\hspace{-0.3cm}&&\hspace{-0.3cm}+\int_{a}^{c}\mathrm{d}y\int_{y-a}^{\infty}\widetilde{\upsilon}(\mathrm{d}\theta)\left(
\Omega^{(q,q+\lambda)}_{\mathbb{W}}(a,b,x,z)\frac{\mathbb{W}_{q+\lambda}(c-y)}{\mathbb{W}_{q+\lambda}(c-a)}
-\Omega^{(q,q+\lambda)}_{\mathbb{W}}(y,b,x,z)\right)
\nonumber\\
\hspace{-0.3cm}&&\hspace{-0.3cm}
+\int_{b}^{z}\mathrm{d}y\int_{y-a}^{\infty}
\upsilon(\mathrm{d}\theta)\bigg(\frac{W_{q}(z-y)}{W_{q}(z-b)}W_{q}(x-b)
-W_{q}(x-y)\bigg).
\end{eqnarray}
By \eqref{r.m.} one has
\begin{eqnarray}\label{r.m.0}
\hspace{-0.3cm}&&\hspace{-0.3cm}
\lambda\int_{a}^{c}\bigg[\frac{\mathbb{W}_{q+\lambda}(x-a)\mathbb{W}_{q+\lambda}(c-y)}
{\mathbb{W}_{q+\lambda}(c-a)}-\mathbb{W}_{q+\lambda}(x-y)\bigg]\mathrm{d}y
\nonumber\\
\hspace{-0.3cm}&=&\hspace{-0.3cm}
\lambda\int_{0}^{\infty}\mathrm{e}^{-(q+\lambda)t}\int_{a}^{c}\widetilde{\mathbb{P}}_{x}(X_{t}\in\mathrm{d}y, t<\tau_{c}^{+}\wedge\tau_{a}^{-})\mathrm{d}t
\nonumber\\
\hspace{-0.3cm}&=&\hspace{-0.3cm}
\lambda\,\widetilde{\mathbb{E}}_{x}\bigg[\int_{0}^{\tau_{c}^{+}\wedge\tau_{a}^{-}}\mathrm{e}^{-(q+\lambda)t}\mathrm{d}t\bigg]
\nonumber\\
\hspace{-0.3cm}&=&\hspace{-0.3cm}
\frac{\lambda}{q+\lambda}\bigg[1-\widetilde{\mathbb{E}}_{x}\bigg[\mathrm{e}^{-(q+\lambda)\tau_{c}^{+}}\mathbf{1}_{\{\tau_{c}^{+}<\tau_{a}^{-}\}}\bigg]-
\widetilde{\mathbb{E}}_{x}\bigg[\mathrm{e}^{-(q+\lambda)\tau_{a}^{-}}\mathbf{1}_{\{\tau_{a}^{-}<\tau_{c}^{+}\}}\bigg]\bigg]
\nonumber\\
\hspace{-0.3cm}&=&\hspace{-0.3cm}
\frac{\lambda}{q+\lambda}\bigg[1-
\frac{\mathbb{W}_{q+\lambda}(x-a)}{\mathbb{W}_{q+\lambda}(c-a)}
-
\bigg[\mathbb{Z}_{q+\lambda}(x-a)-\mathbb{Z}_{q+\lambda}(c-a)\frac{\mathbb{W}_{q+\lambda}(x-a)}{\mathbb{W}_{q+\lambda}(c-a)}\bigg]\bigg].
\end{eqnarray}
Replacing $x$ with $X_{\tau_{b}^{-}}$ in \eqref{r.m.0}, multiplying the resulting equation with $\mathrm{e}^{-q\tau_{b}^{-}}\mathbf{1}_
{\{\tau_{b}^{-}<\tau_{z}^{+}\}}$ and then taking expectation under the measure $\mathbb{P}_{x}$ leads to
\begin{eqnarray}
\label{crt.2}
\hspace{-0.3cm}&&\hspace{-0.3cm}
\lambda\int_{a}^{c}\bigg[\Omega^{(q,q+\lambda)}_{\mathbb{W}}(a,b,x,z)\frac{\mathbb{W}_{q+\lambda}(c-y)}
{\mathbb{W}_{q+\lambda}(c-a)}-\Omega^{(q,q+\lambda)}_{\mathbb{W}}(y,b,x,z)\bigg]\mathrm{d}y
\nonumber\\
\hspace{-0.3cm}&=&\hspace{-0.3cm}
\frac{\lambda}{q+\lambda}\bigg[
Z_{q}(x-b)-Z_{q}(z-b)\frac{W_{q}(x-b)}{W_{q}(z-b)}\bigg]-\frac{\lambda}{q+\lambda}
\frac{\Omega^{(q,q+\lambda)}_{\mathbb{W}}(a,b,x,z)}
{\mathbb{W}_{q+\lambda}(c-a)}
\nonumber\\
\hspace{-0.3cm}&&\hspace{-0.3cm}
-\frac{\lambda}{q+\lambda}\bigg[\Omega^{(q,q+\lambda)}_{\mathbb{Z}}(a,b,x,z)-\frac{\mathbb{Z}_{q+\lambda}(c-a)}
{\mathbb{W}_{q+\lambda}(c-a)}\Omega^{(q,q+\lambda)}_{\mathbb{W}}(a,b,x,z)\bigg].
\end{eqnarray}
Combining \eqref{crt.1}, \eqref{crt.2} and \eqref{g.s.ff}, one obtains \eqref{lap.tran.}. The proof is complete.
\end{proof}

\begin{rem}
\label{two.sid.sol.}
By Corollary \ref{theorem1} one can find that
\begin{eqnarray}
\mathbb{P}_{x}\left(\zeta_{z}^{+}<T\right)
\hspace{-0.3cm}&=&\hspace{-0.3cm}\frac{W_{}(x-b)}{W_{}(z-b)}
+
\frac{\Omega^{(0,\lambda)}_{\mathbb{W}}(a,b,x,z)}{\mathbb{W}_{\lambda}(c-a)-
\Omega^{(0,\lambda)}_{\mathbb{W}}(a,b,c,z)}
\frac{W_{}(c-b)}{W_{}(z-b)},
\nonumber
\end{eqnarray}
from which one can get the solution to the two-side exit problem as
\begin{eqnarray}
\mathbb{E}_{x}\Big(\mathrm{e}^{-q\zeta_{z}^{+}}\mathbf{1}_{\{\zeta_{z}^{+}<T\}}\Big)
\hspace{-0.3cm}&=&\hspace{-0.3cm}
\mathbb{P}_{x}\Big(\zeta_{z}^{+}<T, \zeta_{z}^{+}<e_{q}\Big)
\nonumber\\
\hspace{-0.3cm}&=&\hspace{-0.3cm}
\frac{W_{q}(x-b)}{W_{q}(z-b)}
+
\frac{\Omega^{(q,q+\lambda)}_{\mathbb{W}}(a,b,x,z)}{\mathbb{W}_{q+\lambda}(c-a)-
\Omega^{(q,q+\lambda)}_{\mathbb{W}}(a,b,c,z)}
\frac{W_{q}(c-b)}{W_{q}(z-b)},\nonumber
\end{eqnarray}
where $e_{q}$ is a killing exponential random variable with rate $q$, and is independent of $X$ and $\widetilde{X}$.
\end{rem}

\medskip
By letting $q=0$ and $z\rightarrow \infty$ in Corollary \ref{theorem1}, the following Corollary \ref{cor1} gives the liquidation probability.
\begin{cor} \label{cor1}
For $\lambda>0$ and $a<b<c$ and $b<x$, then
\begin{eqnarray}\label{l.p.}
\mathbb{P}_{x}(T<\infty)\hspace{-0.3cm}&=&\hspace{-0.3cm}1-\psi^{\prime}(0+)\left(W(x-b)+W(c-b)\frac{\Omega^{(0,\lambda)}_{\mathbb{W}}(a,b,x)}{\mathbb{W}_{\lambda}(c-a)-
\Omega^{(0,\lambda)}_{\mathbb{W}}(a,b,c)}\right),
\end{eqnarray}
with $\Omega^{(0,\lambda)}_{\mathbb{W}}(a,b,x)$ being given by \eqref{om.w.}.
\end{cor}

\begin{proof}
Letting $q=0$, and $z\rightarrow \infty$ in Equation \eqref{lap.tran.}, we have
\begin{eqnarray}\label{lap.tran.0}
\mathbb{P}_{x}(T<\infty)
\hspace{-0.3cm}&=&\hspace{-0.3cm}
\lim_{z\rightarrow \infty}\left[1-\frac{W
(x-b)}{W(z-b)}-\frac{\Omega^{(0,\lambda)}_{\mathbb{W}}(a,b,x,z)}{\mathbb{W}_{\lambda}(c-a)}
\right.
\nonumber\\
\hspace{-0.3cm}&&\hspace{-0.3cm}
\left.+\frac{\Omega^{(0,\lambda)}_{\mathbb{W}}(a,b,x,z)}{\mathbb{W}_{\lambda}(c-a)
-\Omega^{(0,\lambda)}_{\mathbb{W}}(a,b,c,z)}
\left(1-\frac{W
(c-b)}{W(z-b)}-\frac{\Omega^{(0,\lambda)}_{\mathbb{W}}(a,b,c,z)}{\mathbb{W}_{\lambda}(c-a)}\right)\right]
\nonumber\\
\hspace{-0.3cm}&=&\hspace{-0.3cm}
\lim_{z\rightarrow \infty}\left[1-\frac{W
(x-b)}{W(z-b)}
-\frac{\Omega^{(0,\lambda)}_{\mathbb{W}}(a,b,x,z)}{\mathbb{W}_{\lambda}(c-a)
-\Omega^{(0,\lambda)}_{\mathbb{W}}(a,b,c,z)}
\frac{W(c-b)}{W(z-b)}\right]
\nonumber\\
\hspace{-0.3cm}&=&\hspace{-0.3cm}
1-\psi^{\prime}(0+)\left(W(x-b)+W(c-b)\frac{\Omega^{(0,\lambda)}_{\mathbb{W}}(a,b,x)}{\mathbb{W}_{\lambda}(c-a)-
\Omega^{(0,\lambda)}_{\mathbb{W}}(a,b,c)}\right),
\end{eqnarray}
where in the last identity we have used $\lim\limits_{x\rightarrow \infty}W(x)=\frac{1}{\psi'(0+)}$ given that $\psi'(0+)>0$.
\end{proof}


\section{Application of the results to the case $X\equiv\widetilde{X}$}
\setcounter{section}{4} \setcounter{equation}{0}

In this section, we recover several existing results on Parisian ruin for a spectrally negative L\'evy risk process with a constant lower barrier in the literature and obtain some new results. To this end, let $X\equiv \widetilde{X}$, hence $\gamma=\widetilde{\gamma}$, $\sigma=\widetilde{\sigma}$, $\upsilon=\widetilde{\upsilon}$, $W_{q}=\mathbb{W}_{q}$ and $Z_{q}=\mathbb{Z}_{q}$ for all $q\geq0$.
It is interesting to see that in this case all results can be expressed by the scale functions $W_{q}$ and $Z_{q}$, while the L\'evy triplet is not needed.

\medskip
The following Corollary \ref{theorem11} gives the Laplace transform of the liquidation time with exponential rehabilitation delay conditioning on that liquidation occurs prior to its first up-crossing of the level $z$.
\begin{cor} \label{theorem11}
For $q,\lambda\in(0,\infty)$, $a<b<c<z$ and $b<x\leq z$, we have
\begin{eqnarray}\label{lap.t.1.}
\mathbb{E}_{x}\bigg[\mathrm{e}^{-qT}\mathbf{1}_{\{T<\zeta^{+}_{z}\}}\bigg]
\hspace{-0.3cm}&=&\hspace{-0.3cm}
K^{(q,q+\lambda)}(a,b,x,z)+\frac{\Omega^{(q,q+\lambda)}_{W}(a,b,x,z)}{W_{q+\lambda}(c-a)
-\Omega^{(q,q+\lambda)}_{W}(a,b,c,z)} K^{(q,q+\lambda)}(a,b,c,z),\nonumber
\end{eqnarray}
where $\Omega^{(q,q+\lambda)}_{W}(a,b,x,z)$ 
is given by \eqref{lap.t.11.},
and
\begin{eqnarray}
K^{(q,q+\lambda)}(a,b,x,z)\hspace{-0.3cm}&=&\hspace{-0.3cm}\frac{q}{q+\lambda}\left[Z_{q+\lambda}(x-a)-\lambda\int_{b}^{x}W_{q}(x-y)Z_{q+\lambda}(y-a)\mathrm{d}y\right.\nonumber\\
\hspace{-0.3cm}&&\hspace{-2cm}
\left.-\frac{W_{q}(x-b)}{W_{q}(z-b)}\left(Z_{q+\lambda}(z-a)-\lambda\int_{b}^{z}W_{q}(z-y)Z_{q+\lambda}(y-a)\mathrm{d}y\right)
\right.\nonumber\\
\hspace{-0.3cm}&&\hspace{-2cm}
\left.-\frac{Z_{q+\lambda}(c-a)}
{W_{q+\lambda}(c-a)}\left[W_{q+\lambda}(x-a)-\lambda\int_{b}^{x}W_{q}(x-y)W_{q+\lambda}(y-a)\mathrm{d}y\right.\right.\nonumber\\
\hspace{-0.3cm}&&\hspace{-2cm}
\left.\left.-\frac{W_{q}(x-b)}{W_{q}(z-b)}\left(W_{q+\lambda}(z-a)-\lambda\int_{b}^{z}W_{q}(z-y)W_{q+\lambda}(y-a)\mathrm{d}y\right)\right]
\right]
\nonumber\\
\hspace{-0.3cm}&&\hspace{-2cm}
\,\,+\frac{\lambda}{q+\lambda}\left[Z_{q}(x-b)-\frac{Z_{q}(z-b)}{W_{q}(z-b)}W_{q}(x-b)\right.
\nonumber\\
\hspace{-0.3cm}&&\hspace{-2cm}
\left.\,-\frac{1}{W_{q+\lambda}(c-a)}\left[W_{q+\lambda}(x-a)-\lambda\int_{b}^{x}W_{q}(x-y)W_{q+\lambda}(y-a)\mathrm{d}y\right.\right.\nonumber\\
\hspace{-0.3cm}&&\hspace{-2cm}
\left.\left.-\frac{W_{q}(x-b)}{W_{q}(z-b)}\left(W_{q+\lambda}(z-a)-\lambda\int_{b}^{z}W_{q}(z-y)W_{q+\lambda}(y-a)\mathrm{d}y
\right)\right]\right].\nonumber
\end{eqnarray}
\end{cor}

\begin{proof}
Note that $W_{q}\equiv \mathbb{W}_{q}$ and $Z_{q}\equiv \mathbb{Z}_{q}$ for all $q\geq0$. The result is a direct consequence of Corollary \ref{theorem1} and \eqref{lap.t.11.}.
\end{proof}

\medskip
In the sequel, we further let $b=0$ and $c\downarrow0$, hence
$$T\rightarrow T_{\lambda}(a)=\tau_{a}^{-}\wedge \tau_{\lambda},
\,\,\,\,
\tau_{\lambda}:=\inf\{t>0; t-g_{t}>e_{\lambda}^{g_{t}},X_{t}<0\},$$
where $g_{t}:=\sup\{s\leq t; X_{s}>0\}$ and $e_{\lambda}^{g_{t}}$ is an exponentially distributed random variable (with parameter $\lambda$) which is independent of $X$. Here, $\tau_{\lambda}$ is the Parisian ruin time with exponential implementation delays, see for example in Landriault et al. (2014) and Baurdoux et al. (2016), while $T_{\lambda}(a)$ is the Parisian ruin time with exponential implementation delays and a lower barrier, see for example in Frostig and Keren-Pinhasik (2019).

\medskip
The following Corollary \ref{KPcor.} gives the two-side  exit identities involving the Parisian ruin time with exponential rehabilitation delay and a lower barrier for the spectrally negative L\'evy processes with either bounded or unbounded variation. It generalizes Proposition 3.1 and Theorems 3.1-3.2 in Frostig and Keren-Pinhasik (2019) as well as Proposition 2.1 of Landriault et al. (2014) where a spectrally negative L\'evy process with bounded variation was considered.
\begin{cor}\label{KPcor.}
For $q,\lambda\in(0,\infty)$, $a\in(-\infty,0)$ and $z\in(0,\infty)$ we have
\begin{eqnarray}\label{11.00}
\hspace{-0.3cm}\mathbb{E}_{x}\left[\mathrm{e}^{-qT_{\lambda}(a)};T_{\lambda}(a)<\tau^{+}_{z}\right]
\hspace{-0.3cm}&=&\hspace{-0.3cm}
\frac{q}{q+\lambda}\left(\ell^{(q,q+\lambda)}(-a,x)-\frac{\ell^{(q,q+\lambda)}(-a,z)}
{\omega^{(q,q+\lambda)}(-a,z)}\omega^{(q,q+\lambda)}(-a,x)\right)\nonumber\\
\hspace{-0.3cm}&&\hspace{-0.3cm}
+\frac{\lambda}{q+\lambda}\left(Z_{q}(x)-\frac{Z_{q}(z)}{\omega^{(q,q+\lambda)}(-a,z)}
\omega^{(q,q+\lambda)}(-a,x)\right), \quad x\in(0,z],
\end{eqnarray}
and
\begin{eqnarray}\label{1x.00}
\mathbb{E}_{x}\Big(\mathrm{e}^{-q\tau_{z}^{+}}\mathbf{1}_{\{\tau_{z}^{+}<T_{\lambda}(a)\}}\Big)
\hspace{-0.3cm}&=&\hspace{-0.3cm}
\frac{\omega^{(q,q+\lambda)}(-a,x)}{\omega^{(q,q+\lambda)}(-a,z)}, \quad x\in(0,z].
\end{eqnarray}
\end{cor}
\begin{proof}
Note that $\mathbb{W}\equiv W$ when $X\equiv \widetilde{X}$. By Lemma 2.2 and (19) of Loeffen (2014) we have
\begin{eqnarray}\label{ome}
\Omega^{(q,q+\lambda)}_{\mathbb{W}}(a,0,x,z)\hspace{-0.3cm}&=&\hspace{-0.3cm}\mathbb{E}_{x}\left[\mathrm{e}^{-q\tau^{-}_{0}}W_{q+\lambda}(X_{\tau^{-}_{0}}-a)
\mathbf{1}_{\{\tau^{-}_{0}<\tau^{+}_{z}\}}\right]\nonumber\\
\hspace{-0.3cm}&=&\hspace{-0.3cm}
\omega^{(q,q+\lambda)}(-a,x)-\frac{W_{q}(x)}{W_{q}(z)}\omega^{(q,q+\lambda)}(-a,z).
\end{eqnarray}
By the definition of $\ell^{(q,q+\lambda)}$ and $\omega^{(q,q+\lambda)}$ it is seen that
\begin{eqnarray}\label{kx}
K^{(q,q+\lambda)}(a,0,x,z)
\hspace{-0.3cm}&=&\hspace{-0.3cm}\frac{q}{q+\lambda}\left[\ell^{(q,q+\lambda)}(-a,x)-\frac{W_{q}(x)}{W_{q}(z)}
\ell^{(q,q+\lambda)}(-a,z)\right.\nonumber\\
\hspace{-0.3cm}&&\hspace{-0.3cm}\left.-\frac{Z_{q+\lambda}(c-a)}{W_{q+\lambda}(c-a)}\left(\omega^{(q,q+\lambda)}(-a,x)-\frac{W_{q}(x)}{W_{q}(z)}
\omega^{(q,q+\lambda)}(-a,z)\right)\right]\nonumber\\
\hspace{-0.3cm}&&\hspace{-0.3cm}
+\frac{\lambda}{q+\lambda}\left[Z_{q}(x)-\frac{Z_{q}(z)}{W_{q}(z)}W_{q}(x)-\frac{1}{W_{q+\lambda}(c-a)}\left(\omega^{(q,q+\lambda)}(-a,x)\right.\right.\nonumber\\
\hspace{-0.3cm}&&\hspace{-0.3cm}\left.\left.-\frac{W_{q}(x)}{W_{q}(z)}
\omega^{(q,q+\lambda)}(-a,z)\right)\right], \quad x\in(0,\infty).\nonumber
\end{eqnarray}
Hence
\begin{eqnarray}\label{lim.c0}
\hspace{-0.3cm}&&\hspace{-0.3cm}
\lim_{c\downarrow0}\frac{\Omega^{(q,q+\lambda)}_{W}(a,0,x,z)}{W_{q+\lambda}(c-a)
-\Omega^{(q,q+\lambda)}_{W}(a,0,c,z)}K^{(q,q+\lambda)}(a,0,c,z)\nonumber\\
\hspace{-0.3cm}&=&\hspace{-0.3cm}
\lim_{c\downarrow0}\frac{\omega^{(q,q+\lambda)}(-a,x)-\frac{W_{q}(x)}{W_{q}(z)}
\omega^{(q,q+\lambda)}(-a,z)}{\frac{W_{q+\lambda}(c-a)-\omega^{(q,q+\lambda)}(-a,c)}{W_{q}(c)}+\frac{\omega^{(q,q+\lambda)}(-a,z)}{W_{q}(z)}
}\left[\frac{q}{q+\lambda}\left(\frac{\ell^{(q,q+\lambda)}(-a,c)}{W_{q}(c)}-\frac{\ell^{(q,q+\lambda)}(-a,z)}{W_{q}(z)}\right.\right.\nonumber\\
\hspace{-0.3cm}&&\hspace{-0.3cm}
\left.\left.-\frac{Z_{q+\lambda}(c-a)}{W_{q+\lambda}(c-a)}\left(\frac{\omega^{(q,q+\lambda)}(-a,c)}{W_{q}(c)}-\frac{\omega^{(q,q+\lambda)}(-a,z)}{W_{q}(z)}
\right)\right)\right.\nonumber\\
\hspace{-0.3cm}&&\hspace{-0.3cm}
\left.+\frac{\lambda}{q+\lambda}\left(\frac{Z_{q}(c)}{W_{q}(c)}-\frac{Z_{q}(z)}{W_{q}(z)}-\frac{1}{W_{q+\lambda}(c-a)}
\left(\frac{\omega^{(q,q+\lambda)}(-a,c)}{W_{q}(c)}-\frac{\omega^{(q,q+\lambda)}(-a,z)}{W_{q}(z)}
\right)\right)\right],\nonumber
\end{eqnarray}
which together with the facts that
\begin{eqnarray}\label{lim.c1}
\hspace{-0.3cm}&&\hspace{-0.3cm}
\left|\frac{\ell^{(q,q+\lambda)}(-a,c)-\frac{Z_{q+\lambda}(c-a)}{W_{q+\lambda}(c-a)}\omega^{(q,q+\lambda)}(-a,c)}{W_{q}(c)}\right|
\nonumber\\
\hspace{-0.3cm}&\leq&\hspace{-0.3cm}
\frac{\lambda\int_{0}^{c}Z_{q+\lambda}(z-a)W_{q}(c-z)\mathrm{d}z}{W_{q}(c)}
+\frac{\lambda Z_{q+\lambda}(c-a)}{W_{q+\lambda}(c-a)}\frac{\int_{0}^{c}W_{q+\lambda}(z-a)W_{q}(c-z)\mathrm{d}z}{W_{q}(c)}
\nonumber\\
\hspace{-0.3cm}&\leq&\hspace{-0.3cm}
\lambda\int_{0}^{c}Z_{q+\lambda}(z-a)\mathrm{d}z
+\frac{\lambda Z_{q+\lambda}(c-a)}{W_{q+\lambda}(c-a)}\int_{0}^{c}W_{q+\lambda}(z-a)\mathrm{d}z
\rightarrow0,\quad c\rightarrow 0^{+},\nonumber
\end{eqnarray}
and
\begin{eqnarray}\label{lim.c2}
\hspace{-0.3cm}&&\hspace{-0.3cm}
\left|\frac{Z_{q}(c)-\frac{\omega^{(q,q+\lambda)}(-a,c)}{W_{q+\lambda}(c-a)}}{W_{q}(c)}\right|
\leq
\frac{q\int_{0}^{c}W_{q}(z)\mathrm{d}z}{W_{q}(c)}
+\frac{\lambda}{W_{q+\lambda}(c-a)}\int_{0}^{c}W_{q+\lambda}(z-a)\mathrm{d}z
\rightarrow0,\quad c\rightarrow 0^{+},\nonumber
\end{eqnarray}
and
\begin{eqnarray}\label{c.0.lim.}
\left|\frac{W_{q+\lambda}(c-a)-\omega^{(q,q+\lambda)}(-a,c)}{W_{q}(c)}\right|
\hspace{-0.3cm}&=&\hspace{-0.3cm}\frac{\lambda\int_{0}^{c}W_{q+\lambda}(z-a)W_{q}(c-z)\mathrm{d}z}{W_{q}(c)}
\nonumber\\
\hspace{-0.3cm}&\leq&\hspace{-0.3cm}
\lambda\int_{0}^{c}W_{q+\lambda}(z-a)\mathrm{d}z
\rightarrow 0,\quad  c\rightarrow 0^{+},
\end{eqnarray}
yields
\begin{eqnarray}\label{lim.c00}
\hspace{-0.3cm}&&\hspace{-0.3cm}
\lim_{c\downarrow0}\frac{\Omega^{(q,q+\lambda)}_{W}(a,0,x,z)}{W_{q+\lambda}(c-a)-\Omega^{(q,q+\lambda)}_{W}(a,0,c,z)}K^{(q,q+\lambda)}(a,0,c,z)\nonumber\\
\hspace{-0.3cm}&=&\hspace{-0.3cm}\left(\frac{W_{q}(z)\omega^{(q,q+\lambda)}(-a,x)}{\omega^{(q,q+\lambda)}(-a,z)}-W_{q}(x)\right)
\left[\frac{q}{q+\lambda}\left(-\frac{\ell^{(q,q+\lambda)}(-a,z)}{W_{q}(z)}+\frac{Z_{q+\lambda}(-a)}{W_{q+\lambda}(-a)}
\frac{\omega^{(q,q+\lambda)}(-a,z)}{W_{q}(z)}\right)\right.\nonumber\\
\hspace{-0.3cm}&&\hspace{-0.3cm}
\left.+\frac{\lambda}{q+\lambda}\left(-\frac{Z_{q}(z)}{W_{q}(z)}+\frac{1}{W_{q+\lambda}(-a)}\frac{\omega^{(q,q+\lambda)}(-a,z)}{W_{q}(z)}
\right)\right]\nonumber\\
\hspace{-0.3cm}&=&\hspace{-0.3cm}
\frac{q}{q+\lambda}\left(-\frac{\ell^{(q,q+\lambda)}(-a,z)}{\omega^{(q,q+\lambda)}(-a,z)}\omega^{(q,q+\lambda)}(-a,x)
+\frac{\ell^{(q,q+\lambda)}(-a,z)}{W_{q}(z)}W_{q}(x)\right.\nonumber\\
\hspace{-0.3cm}&&\hspace{-0.3cm}
\left.+\frac{Z_{q+\lambda}(-a)}{W_{q+\lambda}(-a)}
\omega^{(q,q+\lambda)}(-a,x)-\frac{Z_{q+\lambda}(-a)}{W_{q+\lambda}(-a)}\frac{W_{q}(x)}{W_{q}(z)}
\omega^{(q,q+\lambda)}(-a,z)\right)\nonumber\\
\hspace{-0.3cm}&&\hspace{-0.3cm}
+\frac{\lambda}{q+\lambda}\left(-\frac{Z_{q}(z)}{\omega^{(q,q+\lambda)}(-a,z)}\omega^{(q,q+\lambda)}(-a,x)
+\frac{Z_{q}(z)}{W_{q}(z)}W_{q}(x)\right.\nonumber\\
\hspace{-0.3cm}&&\hspace{-0.3cm}
\left.+\frac{\omega^{(q,q+\lambda)}(-a,x)}{W_{q+\lambda}(-a)}
-\frac{\omega^{(q,q+\lambda)}(-a,z)}{W_{q+\lambda}(-a)}\frac{W_{q}(x)}{W_{q}(z)}
\right),\nonumber
\end{eqnarray}
which implies \eqref{11.00}.
The fluctuation identity \eqref{1x.00} can be recovered as follows. By Remark \ref{two.sid.sol.} and \eqref{ome} one has
\begin{eqnarray}
\mathbb{E}_{x}\Big(\mathrm{e}^{-q\tau_{z}^{+}}\mathbf{1}_{\{\tau_{z}^{+}<T_{\lambda}(a)\}}\Big)
\hspace{-0.3cm}&=&\hspace{-0.3cm}
\frac{W_{q}(x)}{W_{q}(z)}+
\lim_{c\downarrow0}\frac{\omega^{(q,q+\lambda)}(-a,x)-\frac{W_{q}(x)}{W_{q}(z)}\omega^{(q,q+\lambda)}(-a,z)}{W_{q+\lambda}(c-a)
-\omega^{(q,q+\lambda)}(-a,c)+\frac{W_{q}(c)}{W_{q}(z)}\omega^{(q,q+\lambda)}(-a,z)}\frac{W_{q}(c)}{W_{q}(z)}\nonumber\\
\hspace{-0.3cm}&=&\hspace{-0.3cm}
\frac{W_{q}(x)}{W_{q}(z)}+\lim_{c\downarrow0}
\frac{\omega^{(q,q+\lambda)}(-a,x)-\frac{W_{q}(x)}{W_{q}(z)}\omega^{(q,q+\lambda)}(-a,z)}{\frac{W_{q+\lambda}(c-a)
-\omega^{(q,q+\lambda)}(-a,c)}{W_{q}(c)}W_{q}(z)+\omega^{(q,q+\lambda)}(-a,z)}\nonumber\\
\hspace{-0.3cm}&=&\hspace{-0.3cm}
\frac{\omega^{(q,q+\lambda)}(-a,x)}{\omega^{(q,q+\lambda)}(-a,z)},\nonumber
\end{eqnarray}
where in the last identity we have used \eqref{c.0.lim.}.
The proof is complete.
\end{proof}

The following Corollary \ref{cor.4.}, which characterizes the Gerber-Shiu distribution at Parisian ruin time with exponential delay given that the Parisian ruin occurs prior to the first up-crossing time of the level $z$ and the first down-crossing time of the level $a$, recovers Theorem 1.2 of Baurdoux et al. (2016).
\begin{cor}\label{cor.4.}
For $q\geq 0$, $\lambda>0$, $a<0$ and $x\in(0,\infty)$, we have
\begin{eqnarray}\label{l.d.}
\hspace{-0.3cm}&&\hspace{-0.3cm}
\mathbb{E}_{x}\left[\mathrm{e}^{-q\tau_{\lambda}}; X_{\tau_{\lambda}}\in\mathrm{d}u, \tau_{\lambda}<\tau^{-}_{a}\wedge\tau^{+}_{z}\right]\nonumber\\
\hspace{-0.3cm}&=&\hspace{-0.3cm}\lambda\left[\frac{\omega^{(q,q+\lambda)}(-a,x)}{\omega^{(q,q+\lambda)}(-a,z)}\omega^{(q,q+\lambda)}(-u,z)-\omega^{(q,q+\lambda)}
(-u,x)\right]\mathrm{d}u,\quad u\in(a,0],\, z\in(x,\infty).
\end{eqnarray}

\end{cor}

\begin{proof}
Differentiating both sides of \eqref{13} with respect to $u$ and then letting $b=0$ and $c\downarrow0$ in the resulting equations leads to
\begin{eqnarray}\label{11.0}
\hspace{-0.3cm}&&\hspace{-0.3cm}
\mathbb{E}_{x}\left[\mathrm{e}^{-q\tau_{\lambda}}\mathbf{1}_{\{X_{\tau_{\lambda}}\in\mathrm{d}u\}}\mathbf{1}_{\{\tau_{\lambda}<\tau^{-}_{a}\wedge\tau^{+}_{z}\}}\right]\mathbf{1}_{\{a<u\leq 0\}}
\nonumber\\
\hspace{-0.3cm}&=&\hspace{-0.3cm}
\mathbb{E}_{x}\left[\mathrm{e}^{-q\,(\tau_{a}^{-}\wedge \tau_{\lambda})}\mathbf{1}_{\{X_{\tau_{a}^{-}\wedge \tau_{\lambda}}\in \mathrm{d}u\}}\mathbf{1}_{\{\tau_{\lambda}<\tau_{a}^{-}\}}\mathbf{1}_{\{\tau_{a}^{-}\wedge \tau_{\lambda}<\tau^{+}_{z}\}}\right]
\nonumber\\
\hspace{-0.3cm}&=&\hspace{-0.3cm}
\mathbb{E}_{x}\left[\mathrm{e}^{-q\,(\tau_{a}^{-}\wedge \tau_{\lambda})}\mathbf{1}_{\{X_{\tau_{a}^{-}\wedge \tau_{\lambda}}\in \mathrm{d}u\}}\mathbf{1}_{\{a<u\leq 0\}}\mathbf{1}_{\{\tau_{a}^{-}\wedge \tau_{\lambda}<\tau^{+}_{z}\}}\right]
\nonumber\\
\hspace{-0.3cm}&=&\hspace{-0.3cm}
\mathbb{E}_{x}\left[\mathrm{e}^{-qT_{\lambda}(a)}; U_{T_{\lambda}(a)}\in \mathrm{d}u, T_{\lambda}(a)<\tau^{+}_{z}\right]\mathbf{1}_{\{a<u\leq 0\}}
\nonumber\\
\hspace{-0.3cm}&=&\hspace{-0.3cm}\lim_{c\downarrow 0}\lambda\left(\Omega^{(q,q+\lambda)}_{W}(a,0,x,z)\frac{W_{q+\lambda}(c-u)}
{W_{q+\lambda}(c-a)}-\Omega^{(q,q+\lambda)}_{W}(u,0,x,z)\right)
\mathbf{1}_{\{a<u\leq 0\}}\nonumber\\
\hspace{-0.3cm}&&\hspace{-0.3cm}+\lim_{c\downarrow 0}\lambda\frac{\Omega^{(q,q+\lambda)}_{W}(a,0,x,z)}{W_{q+\lambda}(c-a)
-\Omega^{(q,q+\lambda)}_{W}(a,0,c,z)}\left(\Omega^{(q,q+\lambda)}_{W}(a,0,c,z)\right.\nonumber\\
\hspace{-0.3cm}&&\hspace{-0.3cm}
\left.\times\frac{W_{q+\lambda}(c-u)}
{W_{q+\lambda}(c-a)}-\Omega^{(q,q+\lambda)}_{W}(u,0,c,z)\right)\mathbf{1}_{\{a<u\leq 0\}}.
\end{eqnarray}
By \eqref{ome} and the fact that $\omega^{(q,q+\lambda)}(-a,0)=W_{q+\lambda}(-a)$, we rewrite the first term on the right hand side of \eqref{11.0} as
\begin{eqnarray}\label{d.1.}
\hspace{-0.3cm}&&\hspace{-0.3cm}
\lambda\left[\frac{\left[\omega^{(q,q+\lambda)}(-a,x)-\frac{W_{q}(x)}{W_{q}(z)}\omega^{(q,q+\lambda)}(-a,z)\right]
W_{q+\lambda}(-u)}
{W_{q+\lambda}(-a)}\right.\nonumber\\
\hspace{-0.3cm}&&\hspace{0.3cm}
-\omega^{(q,q+\lambda)}(-u,x)+\frac{W_{q}(x)}{W_{q}(z)}\omega^{(q,q+\lambda)}(-u,z)\Bigg]\mathbf{1}_{\{a<u\leq 0\}},
\end{eqnarray}
and the second term on the right hand side of \eqref{11.0} as
\begin{eqnarray}\label{d.2.}
\hspace{-0.3cm}&&\hspace{-0.3cm}
\lambda \mathbf{1}_{\{a<u\leq 0\}} \lim_{c\downarrow 0}\frac{\omega^{(q,q+\lambda)}(-a,x)-\frac{W_{q}(x)}{W_{q}(z)}\omega^{(q,q+\lambda)}(-a,z)}{\frac{W_{q+\lambda}(c-a)-\omega^{(q,q+\lambda)}(-a,c)}{W_{q}(c)}+
\frac{\omega^{(q,q+\lambda)}(-a,z)}{W_{q}(z)}}
\nonumber\\
\hspace{-0.3cm}&&\hspace{-0.3cm}
\times
\left(\frac{\frac{\omega^{(q,q+\lambda)}(-a,c)}{W_{q+\lambda}(c-a)}W_{q+\lambda}(c-u)-\omega^{(q,q+\lambda)}(-u,c)}{W_{q}(c)}\right.\nonumber\\
\hspace{-0.3cm}&&\hspace{-0.3cm}
-\frac{\omega^{(q,q+\lambda)}(-a,z)W_{q+\lambda}(c-u)}{W_{q}(z)W_{q+\lambda}(c-a)}
+\frac{\omega^{(q,q+\lambda)}(-u,z)}{W_{q}(z)}\Bigg)\mathbf{1}_{\{a<u\leq 0\}}\nonumber\\
\hspace{-0.3cm}&=&\hspace{-0.3cm}
\lambda\Bigg[\frac{W_{q+\lambda}(-u)W_{q}(x)}{W_{q+\lambda}(-a)W_{q}(z)}\omega^{(q,q+\lambda)}(-a,z)
-\frac{W_{q+\lambda}(-u)}{W_{q+\lambda}(-a)}\omega^{(q,q+\lambda)}(-a,x)\nonumber\\
\hspace{-0.3cm}&&\hspace{-0.3cm}
\left.
+\left(\frac{\omega^{(q,q+\lambda)}(-a,x)}{\omega^{(q,q+\lambda)}(-a,z)}-\frac{W_{q}(x)}{W_{q}(z)}\right)\omega^{(q,q+\lambda)}(-u,z)\right],
\end{eqnarray}
where we have used \eqref{c.0.lim.} and the fact that
\begin{eqnarray}
\hspace{-0.3cm}&&\hspace{-0.3cm}
\left|\frac{\frac{\omega^{(q,q+\lambda)}(-a,c)}{W_{q+\lambda}(c-a)}W_{q+\lambda}(c-u)-\omega^{(q,q+\lambda)}(-u,c)}{W_{q}(c)}\right|
\nonumber\\
\hspace{-0.3cm}&\leq&\hspace{-0.3cm}
\frac{W_{q+\lambda}(c-u)}{W_{q+\lambda}(c-a)}\frac{\lambda\int_{0}^{c}W_{q+\lambda}(z-a)W_{q}(c-z)\mathrm{d}z}{W_{q}(c)}+
\frac{\lambda\int_{0}^{c}W_{q+\lambda}(z-u)W_{q}(c-z)\mathrm{d}z}{W_{q}(c)}
\nonumber\\
\hspace{-0.3cm}&\leq&\hspace{-0.3cm}
\frac{W_{q+\lambda}(c-u)}{W_{q+\lambda}(c-a)}\frac{\lambda\int_{0}^{c}W_{q+\lambda}(z-a)W_{q}(c)\mathrm{d}z}{W_{q}(c)}+
\frac{\lambda\int_{0}^{c}W_{q+\lambda}(z-u)W_{q}(c)\mathrm{d}z}{W_{q}(c)}
\nonumber\\
\hspace{-0.3cm}&\rightarrow&\hspace{-0.3cm}0,\quad c\rightarrow 0^{+},
\nonumber
\end{eqnarray}
by the definition of $\omega^{(q,q+\lambda)}$ and the monotonicity of $W_{q}$.
Now, substituting \eqref{d.1.} and \eqref{d.2.} in \eqref{11.0} leads to \eqref{l.d.}. The proof is complete.
\end{proof}

\begin{rem}
In Baurdoux et al. (2016), a methodology heavily rely on the geometric distributional property of the number of certain $\varepsilon$-excursions of the spectrally negative L\'evy processes was adopted to prove \eqref{l.d.}. Here we provide an alternative method to obtain the same conclusion. By the same arguments as that right below Corollary 1.1 in Baurdoux et al. (2016), Proposition 2.1 of Landriault et al. (2014) can also be recovered from our result \eqref{l.d.}.
\end{rem}

\begin{rem}
When $a\downarrow-\infty$ in Corollary \ref{cor.4.}, then for $u\in(-\infty,0)$ and $z\in(x,\infty)$ we have
\begin{eqnarray}
\mathbb{E}_{x}\big[\mathrm{e}^{-q\tau_{\lambda}}; X_{\tau_{\lambda}}\in\mathrm{d}u, \tau_{\lambda}<\tau^{+}_{z}\big]
=
\lambda\Bigg[\frac{\mathcal{H}^{(q+\lambda,-\lambda)}(x)}{\mathcal{H}^{(q+\lambda,-\lambda)}(z)}\omega^{(q,q+\lambda)}(-u,z)-\omega^{(q,q+\lambda)}(-u,x)\Bigg]\mathrm{d}u,\nonumber
\end{eqnarray}
which recovers (1.8) of Baurdoux et al. (2016), and $\mathcal{H}^{(q+\lambda,-\lambda)}(x):=\mathrm{e}^{\Phi_{q+\lambda}x}[1-\lambda\int_{0}^{x}\mathrm{e}^{-\Phi_{q+\lambda}y}
W_{q}(y)\mathrm{d}y]$.
\end{rem}



The following Corollary \ref{cor11} gives the the Parisian ruin
with exponential implementation delays and a lower barrier $a$, which recovers (3.30) in Frostig and Keren-Pinhasik (2019).
\begin{cor} \label{cor11}
For $a<0$ and $x>0$, then
\begin{eqnarray}\label{p.l.1}
\mathbb{P}_{x}(T_{\lambda}(a)<\infty)=1-\psi'(0+)\frac{\omega^{(0,\lambda)}(-a,x)}{
Z_{\lambda}(-a)}.
\end{eqnarray}
\end{cor}

\begin{proof}
Note that $\mathbb{W}_{q}\equiv W_{q}$ for all $q\geq0$. Letting $b=0$ and $c\downarrow0$ in (\ref{l.p.}) gives
\begin{eqnarray}\label{l.p.1.}
\mathbb{P}_{x}(T_{\lambda}(a)<\infty)=1-\psi'(0+)\lim_{c\downarrow0}\left(W(x)+W(c)\frac{\Omega^{(0,\lambda)}_{W}(a,0,x)}{W_{\lambda}(c-a)-
\Omega^{(0,\lambda)}_{W}(a,0,c)}\right).
\end{eqnarray}
It can be verified from \eqref{equi.def.w}, \eqref{s.f.l.}, \eqref{lap.t.11.} and the definition of $\Omega^{(0,\lambda)}_{W}(a,0,x)$ in Corollary \ref{cor1} that
\begin{eqnarray}\label{l.p.o.}
\Omega^{(0,\lambda)}_{W}(a,0,x)\hspace{-0.3cm}&=&\hspace{-0.3cm}\omega^{(0,\lambda)}(-a,x)-W(x)\lim_{z\rightarrow\infty}\left(\frac{W(z-a)}{W(z)}+\lambda\int_{0}^{-a}
\frac{W(z-a-y)}{W(z)}W_{\lambda}(y)\mathrm{d}y\right)\nonumber\\
\hspace{-0.3cm}&=&\hspace{-0.3cm}
\omega^{(0,\lambda)}(-a,x)-W(x)\left(\mathrm{e}^{-\Phi(0)a}+\lambda\int_{0}^{-a}\mathrm{e}^{-\Phi(0)(y+a)}
W_{\lambda}(y)\mathrm{d}y\right)\nonumber\\
\hspace{-0.3cm}&=&\hspace{-0.3cm}
\omega^{(0,\lambda)}(-a,x)-W(x)\left(1+\lambda\int_{0}^{-a}
W_{\lambda}(y)\mathrm{d}y\right)\nonumber\\
\hspace{-0.3cm}&=&\hspace{-0.3cm}
\omega^{(0,\lambda)}(-a,x)-W(x)Z_{\lambda}(-a),\nonumber
\end{eqnarray}
which substituted into (\ref{l.p.1.}) yields
\begin{eqnarray}\label{l.p.11.}
\mathbb{P}_{x}(T_{\lambda}(a)<\infty)\hspace{-0.3cm}&=&\hspace{-0.3cm}1-\psi'(0+)\lim_{c\downarrow0}\left(W(x)+W(c)\frac{\omega^{(0,\lambda)}(-a,x)-W(x)Z_{\lambda}(-a)}{W_{\lambda}(c-a)-
\omega^{(0,\lambda)}(-a,c)+W(c)Z_{\lambda}(-a)}\right)
\nonumber\\
\hspace{-0.3cm}&=&\hspace{-0.3cm}1-\psi'(0+)\lim_{c\downarrow0}\left(W(x)+\frac{\omega^{(0,\lambda)}(-a,x)-W(x)Z_{\lambda}(-a)}{\frac{W_{\lambda}(c-a)-
\omega^{(0,\lambda)}(-a,c)}{W(c)}
+Z_{\lambda}(-a)}\right)\nonumber\\
\hspace{-0.3cm}&=&\hspace{-0.3cm}
1-\psi'(0+)\frac{\omega^{(0,\lambda)}(-a,x)}{
Z_{\lambda}(-a)},
\end{eqnarray}
where in the third identity we have used \eqref{c.0.lim.}.
The proof is complete.
\end{proof}

\begin{rem} \label{rem.7}
Due to \eqref{s.f.l.} one can deduce that
\begin{eqnarray}
\lim\limits_{a\downarrow-\infty}\frac{\omega^{(0,\lambda)}(-a,x)}{Z_{\lambda}(-a)}\hspace{-0.3cm}&=&\hspace{-0.3cm}
\lim\limits_{a\downarrow-\infty}\frac{W_{\lambda}(x-a)}{Z_{\lambda}(-a)}-\lambda\int_{0}^{x}
\bigg[\lim\limits_{a\downarrow-\infty}\frac{W_{\lambda}(z-a)}{Z_{\lambda}(-a)}\bigg]W_{}(x-z)\mathrm{d}z
\nonumber\\
\hspace{-0.3cm}&=&\hspace{-0.3cm}
\frac{\Phi_{\lambda}}{\lambda}\bigg[\mathrm{e}^{\Phi_{\lambda}x}-\lambda\int_{0}^{x}
\mathrm{e}^{\Phi_{\lambda}z}\,W_{}(x-z)\mathrm{d}z\bigg],\nonumber
\end{eqnarray}
which together with Corollary \ref{cor11} gives the following expression of the Parisian ruin probability
\begin{eqnarray}\label{p.p.}
\mathbb{P}_{x}(\tau_{\lambda}<\infty)=\lim\limits_{a\downarrow-\infty}\mathbb{P}_{x}(T_{\lambda}(a)<\infty)=1-\frac{\psi^{\prime}(0+)\Phi_{\lambda}}{\lambda}\bigg[\mathrm{e}^{\Phi_{\lambda}x}-\lambda\int_{0}^{x}
\mathrm{e}^{\Phi_{\lambda}z}\,W_{}(x-z)\mathrm{d}z\bigg],
\end{eqnarray}
which recovers (1.13) of Baurdoux et al. (2016).
\end{rem}

\section{Illustrative Examples}
\setcounter{section}{5} \setcounter{equation}{0}

This section aims to illustrate the the results derived in the previous sections. One quantity of
interest is the probability of liquidation, which includes the probability of Parison ruin as a special case.
The other item is the discounted joint probability distribution of the liquidation time with an exponential rehabilitation delay, the surplus at and the running supreme of the surplus until the liquidation time with exponential rehabilitation delay.

\subsection{A comparison of probabilities of Parisian ruin and Liquidation ruin}

Parisian ruin occurs once the surplus stays continuously below zero for a given period. This is to amend the classical risk theory in the sense that usually companies are not bankrupt when down-crossing a predetermined
low level (e.g. zero). That is, Parisian ruin describe a scenario while the surplus down-crosses zero, but bankruptcy may occur later. In the last decade there is a growing literature on this topic, see Yang et al. (2020), etc.

The liquidation ruin introduced in this paper also carries a similar mission to imitate the real-world process of bankruptcy. This subsection intends to make a visualized comparison between Parisian ruin and Liquidation ruin. To make the comparison fair and meaningful, we need to calculate the two scenarios of ruin under the same parameter setting. As mentioned in Section 4, Parisian only makes sense when $X=\widetilde{X}$, which implies the same underlying surplus process for the states of {\it solvent} and {\it insolvent} in the sense of Liquidation ruin. Hence we do not include the case of dividend distribution in the original surplus process. In this subsection, the calculation of Liquidation ruin is based on the results in Corollary 3, while Parisian ruin is based on Remark \ref{rem.7}.

In the setting of the three-barrier system, when $a>= 0$, the down-crossing of level $a$ must occur before the down-crossing of level 0, hence the liquidation ruin time must occur before the Parisian ruin. When $a<0$ and $b>= 0$, the prolonged sinking duration below level 0 implies a prolonged sinking duration below level $c$, then the liquidation ruin still occurs earlier. To avoid trivial cases, we need to let $a<0$ and $b<0$.

In the setting of the underlying surplus process, we use the same jump-diffusion process $X(t)$ in both the {\it solvent} and {\it insolvent} states,
\[{X(t)} = x + c_0 t + \sigma_0 {B_t} - \sum\limits_{i = 1}^{{N(t)}} {{Y_i}} {\rm{ }},\quad t\geq 0,\]
where $\sigma_0 {\rm{ > 0}}$, $ \left\{ {{N(t)},t \ge 0} \right\}$  is a Poisson process with arrival rate $\lambda_0$, and $Y_{i}$'s  are a sequence of i.i.d. random variables distributed with Erlang $ \left( {2,\alpha } \right).$ The Laplace exponent of $X(t)$ is
$$\psi(\theta)=c_0\theta-\lambda_0+\frac{\lambda_0\alpha^2}{(\alpha+\theta)^2}+\frac{1}{2}\sigma_0^2\theta^2,$$ and hence $\psi'(0+)=c_0-\frac{2\lambda_0}{\alpha}$.
The scale function associated with $X$ can be derived as (cf., Loeffen (2008))
\begin{align}\label{expressionsofWqin2}
{W_{q}}\left( x \right) = \sum\limits_{j = 1}^4 {{D_j}\left( q \right){\mathrm{e}^{{\theta _j}\left( q \right)x}}} ,{\kern 1pt} {\kern 1pt} {\kern 1pt} {\kern 1pt} {\kern 1pt} {\kern 1pt} {\kern 1pt} {\kern 1pt} {\kern 1pt} {\kern 1pt} {\kern 1pt} x \ge 0,
 \end{align}
where
\begin{align}\label{expressionsofdj}
{D_j}\left( q \right) = \frac{{{{\left( {\alpha  + {\theta _j}\left( q \right)} \right)}^2}}}{{\frac{\sigma_0 ^2}{2}\prod\limits_{i = 1,i \ne j}^4 {\left( {{\theta _j}\left( q \right) - {\theta _i}\left( q \right)} \right)} }},
\nonumber
 \end{align}
and ${\theta _j}\left( q \right)\left( {j = 1, \cdots 4} \right)$ are the (distinct) zeros of the polynomial
\begin{align}
\nonumber &\quad\Big(c_0\theta-\lambda_0+\frac{\lambda_0\alpha^2}{(\alpha+\theta)^2}+\frac{1}{2}\sigma_0^2\theta^2-q\Big)(\alpha+\theta)^2\\
\nonumber &=\frac{1}{2}\sigma_0^2\theta^4+(\alpha\sigma_0^2+c_0)\theta^3+\big(\frac{1}{2}\sigma_0^2\alpha^2-\lambda_0-q+2c_0\alpha\big)\theta^2+\big[c_0\alpha^2-2(\lambda_0+q)\alpha\big]\theta-q\alpha^2.
\end{align}

Then the Liquidation ruin probability discussed in this subsection is given in equation \eqref{l.p.} and the Parisian ruin probability is given in \eqref{p.p.} with $\upsilon(\mathrm{d}z)=\lambda_0 \alpha^{2} z \mathrm{e}^{-\alpha z} \mathrm{d}z$. The parameter values in this example are: $c_0=5$, $\sigma_0=0.5$,  $\lambda_0=\alpha=2$, $\lambda=0.1$ and $x=1$. Under this setting, the Parisian ruin is calculated as 0.0138, which is depicted as the blue surface in Figure 2. On the side of liquidation ruin probability, the value depends on the setting of the three-barrier system. We would like to find the impact of the structure of three barriers on the value of liquidation ruin probability.

From mathematical point of view, Parisian ruin is the limit of liquidation ruin when $a=-\infty$ and $b=c$. By varying the levels of three barriers $a<b<c$, the results of liquidation ruin probability is given as the red surface in Figure 2. As the liquidation barrier $a$ is lower, the chance of direct liquidation by hitting $a$ is smaller, then the liquidation ruin probability is lower. From (III) to (I) in Figure 2, the value of $a$ decreases from -3 to -100, we can see the red surface is getting lower and closer to the blue surface.

Another obvious trend in Figure 2 is the impacts of $b$ and $c$. The higher level of $b$, the more likely to be in the state of {\it insolvent}, the higher chance to be liquidated. On the other hand, the higher level of $c$, the less likely to climb back to the safety level, the higher probability of liquidation ruin. We can observe this trend easily in the red surface by varying the values of $b$ and $c$.

By implementing a stricter condition of returning to the status of financially healthy (higher than $c$ rather than $b$), the liquidation ruin is usually believed to be more likely to occur, comparing to Parisian ruin. This refers to the majority cases in Figure 2 with the red surface is higher. Only when the levels of $a$, $b$ and $c$ are low enough, we have a lower liquidation ruin probability.

We are also interested in the case when the probabilities of liquidation ruin and Parisian ruin are equal. From (I) and (II) in Figure 2, the intersection of the red and blue surfaces defines a curve in the coordinate system of $b$ and $c$. This curve is closer to the $b$ axis which indicates the liquidation ruin probability is more sensitive to the change of $c$ than the change of $b$. Besides, as the $a$ level becomes lower from $a=-5$ in (II) to $a=-100$ in (I), the level of the red surface keeps lowering leading to the intersection curve moves to the right hand side.


\begin{figure}[!htp]
\centering{}\includegraphics[width=6.5in,height=4in]{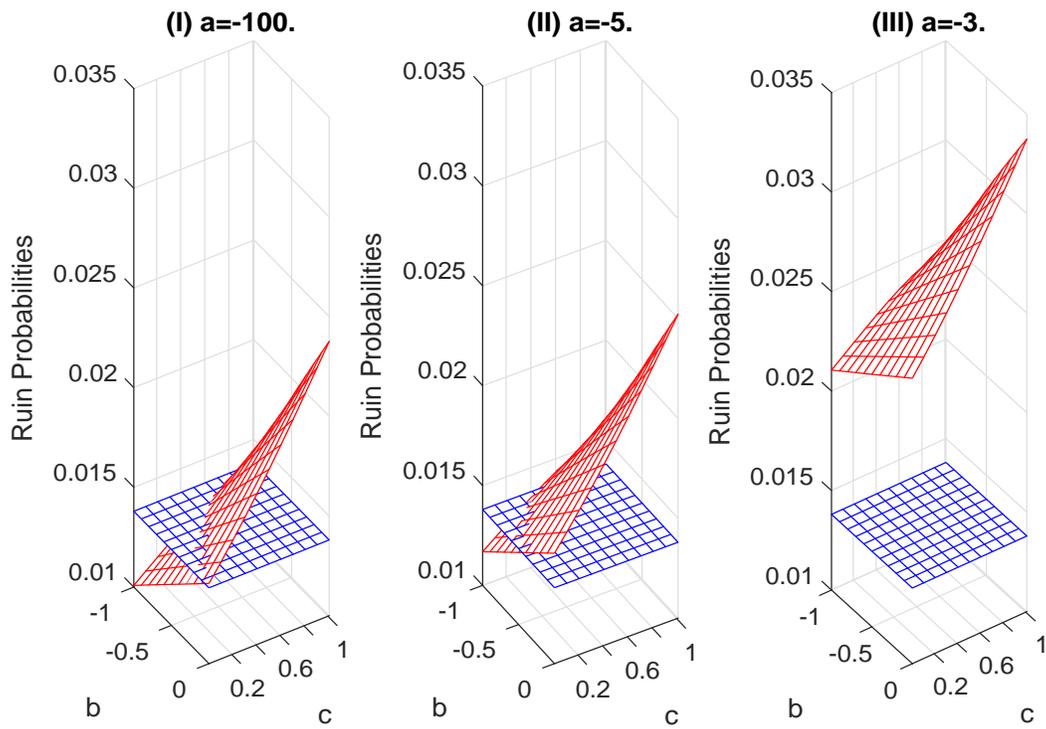}
\caption{A comparison of Parisian ruin probability and Liquidation ruin probability with different barriers.}
\end{figure}

\subsection{The effects of parameters on the Liquidation ruin probability}

The previous example focuses mainly on the impacts of the three-barrier system on the ruin probability. In this subsection, we proceed with the effects of other model parameters on the liquidation ruin probability. We fix the three-barrier system as $a=0$, $b=1$, $c=2$, and $\lambda=0.2$ for the grace period. We no longer consider Parisian ruin in this example and hence we are able to use two different surplus processes in the state of {\it solvent} and {\it insolvent}. Let $X$ be a jump-diffusion processes, i.e.,
\[X_{t} = x + p t + \sigma {B_t} - \sum\limits_{i = 1}^{{N_{1}(t)}} {{Y_i}} - D_{t},\quad t\geq 0,\]
where $B_t$ is a Brownian motion, $p,\sigma>0$, $ \left\{ {{N_{1}(t)};t \ge 0} \right\}$ is a Poisson processes with arrival rate $\lambda_{1}$, $\{Y_{i}; i\geq 1\}$ is a sequence of i.i.d. random variables with Erlang$(2,\beta)$ distribution law $F(\mathrm{d}x)=\beta^{2}x\mathrm{e}^{-\beta x}\mathrm{d}x$ for $x>0$ and $\beta>0$, and
$$D_{t}=\alpha t+\sum\limits_{i = 1}^{{N_{2}(t)}} {Z_i} ,\quad t\geq 0,$$
represents the process of accumulated dividends over the time period $[0,t]$, where $0<\alpha<p$, $ \left\{ {{N_{2}(t)};t \ge 0} \right\}$ is a Poisson processes with arrival rate $\lambda_{2}$, $\{Z_{i}; i\geq 1\}$ is a sequence of i.i.d. random variables with exponential distribution law $G(\mathrm{d}x)=\gamma\mathrm{e}^{-\gamma x}\mathrm{d}x$ for $x>0$ and $\gamma>0$. Here $D_t$ represents the dividend payment whenever the company is in a solvent state. The first term $\alpha t$ refers to the part distributed by the threshold strategy, where we take $c$ as the threshold, then when the company is in the state of solvent, dividends are paid $\alpha$ per unit time. The second term $\sum\limits_{i = 1}^{{N_{2}(t)}} {Z_i}$ refers to the dividend strategy with random observation times described by an exponential distribution. The expression $D_t$ can be used to describe a mixture of the threshold dividend strategy and the impulse dividend strategy by adjusting the parameter $\alpha$ and $\lambda_2$. For literature on risk models embedded with other well known dividend strategies, we are referred to Lin et al. (2003), Loeffen (2008), Avanzi et al. (2013), Yin and Wen (2013), Cheung and Wong (2017), Xu and Woo (2020)  and the references therein.

When the company is in the state of {\it insolvent}, the distribution of dividend is suspended temperately and will not retrieve until the company returns to the state of {\it solvent}. Besides, the premium income, the revenue volatility and the claims frequency will also be affected. Therefore, let
\[{\widetilde{X}(t)} = x + \tilde{p}t + \tilde{\sigma} \tilde{B}_{t} - \sum\limits_{i = 1}^{{\tilde{N}_{1}(t)}} {{\widetilde{Y}_i}},\quad t\geq 0,\]
where $\tilde{p},\tilde{\sigma}>0$, $\tilde{B}_t$ is a Brownian motion, $\{ {{\tilde{N}_{1}(t)};t \ge 0}\}$ is a Poisson processes with arrival rate $\tilde{\lambda}_{1}$, $\{\tilde{Y}_{i}; i\geq 1\}$ is a sequence of i.i.d. random variables with Erlang$(2,\tilde{\beta})$ distribution law $\tilde{F}(\mathrm{d}x)=\tilde{\beta}^{2}x\mathrm{e}^{-\tilde{\beta} x}\mathrm{d}x$ for $x>0$ and $\tilde{\beta}>0$.
The L\'{e}vy measures for $X$ and $\widetilde{X}$ are given by $\upsilon(\mathrm{d}z)=\lambda_1 F(\mathrm{d}z)+\lambda_2 G(\mathrm{d}z)$ and $\widetilde{\upsilon}(\mathrm{d}z)=\tilde{\lambda}_1 \tilde{F}(\mathrm{d}z)$.
The scale functions associated with $X$ and $\widetilde{X}$ can be derived as
\begin{align}\label{expressionW1}
{W_{q}}\left( x \right) = \sum\limits_{j = 1}^5 {{C_j}\left( q \right){\mathrm{e}^{{\theta _j}\left( q \right)x}}} {\kern 1pt} {\kern 1pt} {\kern 1pt} {\kern 1pt} {\kern 1pt} \mathrm{and} {\kern 1pt} {\kern 1pt}{\kern 1pt} {\kern 1pt} {\kern 1pt} {\kern 1pt} {\kern 1pt} {\mathbb{W}_{q}}\left( x \right) = \sum\limits_{j = 1}^4 {{D_j}\left( q \right){\mathrm{e}^{{\eta _j}\left( q \right)x}}} ,{\kern 1pt} {\kern 1pt} {\kern 1pt} {\kern 1pt} {\kern 1pt} {\kern 1pt} {\kern 1pt} {\kern 1pt} {\kern 1pt} x \ge 0,
 \end{align}
and
\begin{align}\label{expressionZ1}
Z_{q}(x)=1+q\sum_{j=1}^{5}\frac{C_{j}(q)}{\theta_{j}(q)}\left(\mathrm{e}^{{\theta _j}\left( q \right)x}-1\right) {\kern 1pt} {\kern 1pt} {\kern 1pt} {\kern 1pt} {\kern 1pt} \mathrm{and} {\kern 1pt} {\kern 1pt}{\kern 1pt} {\kern 1pt} {\kern 1pt}  {\kern 1pt} {\kern 1pt} \mathbb{Z}_{q}(x)=1+q\sum_{j=1}^{4}\frac{D_{j}(q)}{\eta_{j}(q)}\left(\mathrm{e}^{{\eta _j}\left( q \right)x}-1\right),
 \end{align}
where
\begin{align}\label{expressionsofdj}
{C_j}\left( q \right) = \frac{{{{\left( {\beta  + {\theta _j}\left( q \right)} \right)}^2}}\left( {\gamma  + {\theta _j}\left( q \right)} \right)}{{\frac{\sigma ^2}{2}\prod\limits_{i = 1,i \ne j}^5 {\left( {{\theta _j}\left( q \right) - {\theta _i}\left( q \right)} \right)} }}  {\kern 1pt} {\kern 1pt} {\kern 1pt} {\kern 1pt} {\kern 1pt}
\mathrm{and}
{\kern 1pt} {\kern 1pt}{\kern 1pt} {\kern 1pt} {\kern 1pt} {D_j}\left( q \right) = \frac{{{{\left( {\tilde{\beta}  + {\eta _j}\left( q \right)} \right)}^2}}}{{\frac{\tilde{\sigma} ^2}{2}\prod\limits_{i = 1,i \ne j}^4 {\left( {{\eta _j}\left( q \right) - {\eta _i}\left( q \right)} \right)} }},
\nonumber
 \end{align}
and ${\theta _j}\left( q \right) \left(j\leq 5\right) $ and $ {\eta _j}\left( q \right)\left( {j \leq 4} \right)$ are, respectively, the (distinct) zeros of the polynomials
\begin{align}
\nonumber &(\psi(\theta)-q)(\beta+\theta)^2(\gamma+\theta)=\Big(\frac{1}{2}\sigma^2\theta^2+(p-\alpha)\theta-\lambda_1+\frac{\lambda_1\beta^2}{(\beta+\theta)^2}-\frac{\lambda_{2}\theta}{\gamma+\theta}-q\Big)(\beta+\theta)^2(\gamma+\theta),
\end{align}
and
\begin{align}
\nonumber &(\widetilde{\psi}(\theta)-q)(\tilde{\beta}+\theta)^2=\Big(\frac{1}{2}\tilde{\sigma}^2\theta^2+\tilde{p}\theta-\tilde{\lambda}_1+\frac{\tilde{\lambda}_1\tilde{\beta}^2}{(\tilde{\beta}+\theta)^2}-q\Big)(\tilde{\beta}+\theta)^2.
\end{align}
Besides,
\begin{align}
\nonumber &\psi^{\prime}(0)=\left.\psi^{\prime}(\theta)\right|_{\theta=0}
=
p-\alpha-\frac{2\lambda_1}{\beta}-\frac{\lambda_{2}}{\gamma}, \text{ and } \widetilde{\psi}^{\prime}(0)=\left.\widetilde{\psi}^{\prime}(\theta)\right|_{\theta=0}
=
\tilde{p}-\frac{2\tilde{\lambda}_1}{\tilde{\beta}}.
\end{align}

The base parameters in this subsection is given as: $x=5$, $p=10$, $\sigma=2$, $\lambda_1=\tilde{\lambda_1}=3$, $\beta=\tilde{\beta}=2$, $\alpha=2$, $\lambda_2=2$, $\gamma=1$, $\tilde{p}=6$, and $\tilde{\sigma}=1.5$. Then the liquidation ruin probability can be calculated by equation \eqref{l.p.} in Corollary 3. Figure 3 gives the trends of liquidation ruin probability when we vary the initial wealth $x$, the premium rate $p$ in the {\it solvent} state, the parameter $\lambda$ in the exponential distributed grace period, and the dividend parameters $\alpha$, $\lambda_2$ and $\gamma$.

In Figure 3, the impacts of $x$ and $p$ are obvious in (I) and (II). The higher value of $x$ or $p$, the level of surplus of the insurer is higher, which leads to a smaller chance of liquidation ruin. The parameter $\lambda$, however, does not directly affect the level of wealth, but controls the length of the grace period. The higher $\lambda$, the shorter of the grace period, the more pressure for the surplus process to climb up to the safety barrier, then the more likely the insurer be concluded to liquidation for staying overtime at the state of {\it insolvent}. This explains the increasing trend in (III). Regarding to the dividend parameters, higher $\alpha$ gives more dividend distributed to the shareholders, which brings down the surplus level of the insurer and hence an increasing ruin probability in (IV). The parameter $\lambda_2$ affects the frequency to pay dividend: the higher $\lambda_2$ the more frequent the insurer to pay out dividend. Then the surplus level jumps down more frequently which makes it harder to climb up to the safety barrier, and hence a higher ruin probability as seen in (V). The parameter $\gamma$ affects the dividend amount which is smaller as $\gamma$ increases. With more profit retained in the surplus process, the ruin probability of the insurer reduces as shown in (VI). These numerical illustrations indicate that our theoretical results  fits the intuitive explanations well.

\begin{figure}[!htp]

\centering{}\includegraphics[width=6.5in,height=5in]{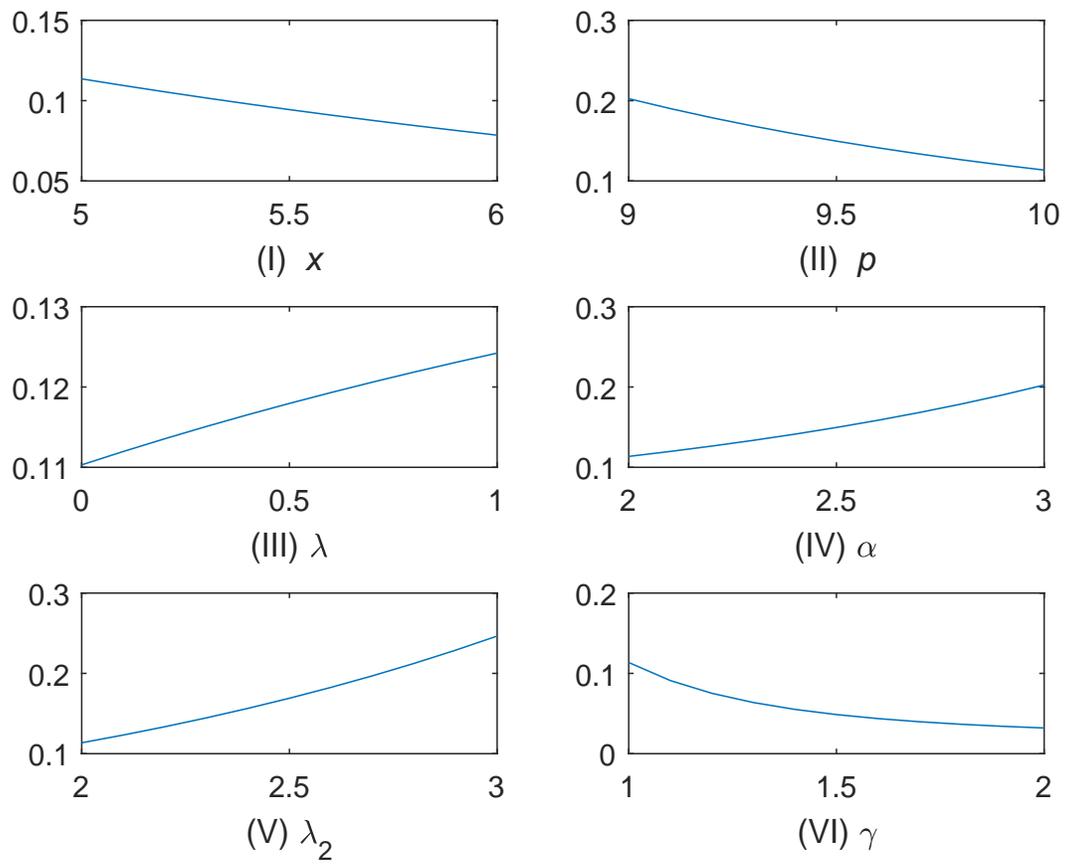}
\caption{The effects of model parameters on the liquidation ruin.}
\end{figure}

\subsection{The joint distribution of the liquidation time, the surplus at liquidation and the historical high of the surplus until liquidation}

For computational simplicity, we use the same model setting as in the previous subsection except $\lambda=0.1$ in this example. The value of the initial wealth $x$ and the discount factor $q$ in (\ref{12}) and (\ref{13}) will also be varied. The three-barrier system remains unchanged as $a=0$, $b=1$ and $c=2$. Instead of ruin probability, in this subsection we study the joint distribution of the liquidation time, the surplus at liquidation and the historical high of the surplus until liquidation. According to Corollary 1, the joint distribution function can be calculated based on (\ref{12}) and (\ref{13}) piece-wisely.

Figure 4 plots the surfaces of the joint distribution in four scenarios of $q$ and $x$. For the convenience of comparison, we adopt the same axis of ordinates [0,0.31]. This range is chosen to reflect the value range in (\ref{12}) and (\ref{13}), where $u\in (-\infty,a)\cup [a,c]$ and $z\in[x, \infty)$. Note that the joint distribution for $z<x$ is always 0 indicating that the historical high of the surplus should be greater than or equal to the initial wealth $x$, and hence there is no need to illustrate the case of $z<x$. Moreover, the lower bound $u=-1$ and the upper bound $z=12$ in our variable ranges are obtained by trail and error until a steady behavior of the surface is clearly shown. The values of the joint distribution based on (\ref{12}) for $u\in[-1,0)$ are plotted in red, while (\ref{13}) for $u\in[0,2]$ leads to the blue surfaces in Figure 4.

Figure 4(I) plots the surface of the distribution for $q=0$ and $x=2$. One notable feature of this surface is the discontinuity at $u=a=0$. In fact, this behavior is the direct result of Corollary \ref{cor.2.} with $q=0$, from which one can see that
\begin{eqnarray}\label{obs1}
\hspace{-0.3cm}&&\hspace{-0.3cm}
\mathbb{P}_{x}\left(U_{T}=a,\overline{U}_{T}< z\right)
=
\frac{\widetilde{\sigma}^2}{2}
\left[\Omega^{(q,q+\lambda)}_{\mathbb{W}'}(a,b,x,z)-
\Omega^{(q,q+\lambda)}_{\mathbb{W}}(a,b,x,z)\frac{\mathbb{W}_{q+\lambda}^{\prime}(c-a)}{\mathbb{W}_{q+\lambda}(c-a)}\right.
\nonumber\\
\hspace{-0.3cm}&&\hspace{-0.3cm}
\quad
\left.+\frac{\Omega^{(q,q+\lambda)}_{\mathbb{W}}(a,b,x,z)}{\mathbb{W}_{q+\lambda}(c-a)
-\Omega^{(q,q+\lambda)}_{\mathbb{W}}(a,b,c,z)}
\left[\Omega^{(q,q+\lambda)}_{\mathbb{W}'}(a,b,c,z)-
\Omega^{(q,q+\lambda)}_{\mathbb{W}}(a,b,c,z)\frac{\mathbb{W}_{q+\lambda}^{\prime}(c-a)}{\mathbb{W}_{q+\lambda}(c-a)}\right]\right],
\end{eqnarray}
which accounts for the jump from the red surface to the blue surface. Theoretically speaking, the jump amount given by \eqref{obs1} in the bi-variate distribution function of $U_{T}$ and $\overline{U}_{T}$ can be explained by the creeping behavior of the Brownian motion part $\widetilde{B}$ in $\widetilde{X}$. That is, the liquidation barrier $a$ can be hit with a positive probability when $\widetilde{X}$ down-crosses $a$ which implies a direct liquidation. Then at the time of liquidation $T$, there is a positive probability that, the insurer's surplus hits $a$, i.e., $U_{T}=a$. On the other hand, if $\widetilde{\sigma}=0$ which brings no Brownian motion part in the process of $\widetilde{X}$, then the jump between the two surfaces no longer exists.

Another interesting feature of Figure 4(I) is that when $z=2$, both the values of the red and blue surfaces are 0, which means that the bi-variate distribution function of $U_{T}$ and $\overline{U}_{T}$ at $z=2$ vanishes for all $u$. Note that we have set the initial wealth $x=2=c>b>a$ in this scenario, which implies the liquidation time $T$ must occur later than time 0. Starting from $x=2$, we have the surplus process enters $(2, \infty)$ immediately with probability one. In fact, when $X$ has paths of unbounded variation, this result can be seen from Theorem 6.5 of Kyprianou (2014, Second Edition) and the associated remarks. In the case when $X$ has paths of bounded variation, $X$ can be decomposed as a positive linear drift minus a pure jump subordinator with bounded variation (see, Lemma 2.12 and Lemma 2.14 of Kyprianou (2014, Second Edition)). Due to the positive linear drift, $X$ enters $(2,\infty)$ immediately, see Page 157 of Kyprianou (2014, Second Edition). Therefore, as long as the liquidation time $T>0$, the running supreme of the surplus process between 0 and $T$ must be greater than 2 almost surely under $\mathbb{P}_{2}$, that is
$$\mathbb{P}_{2}\left(U_{T}\leq u,\overline{U}_{T}< 2\right)\leq \mathbb{P}_{2}\left(\overline{U}_{T}< 2\right)=0,\quad u\in(-\infty,\infty),$$
and hence the boundary value of the surface when $z=2$ in Figure 4(I) is 0.


Besides, the shape of the surface in Figure 4(I) is consistent to its nature as a bi-variate distribution function. For example, as $u$ or $z$ increases, the level of the surface is higher. As $u$ goes to $-\infty$, the level of the surface is lower and approaching 0. When both $u$ and $z$ increase to $+\infty$, the surface level presents a trend of convergence to a constant which is less than 1 in this case. In fact, this is the so-called defective distribution function which can be seen in $\mathbb{E}_{x}\left[\mathbf{1}_{\{U_{T}\leq u\}}\mathbf{1}_{\{\overline{U}_{T}< z\}}\right]$, where the limit is the probability of liquidation with a value of 0.1998 (by Corollary 3) in this scenario.

\begin{figure}[!htp]
\centering{}\includegraphics[width=6in,height=5.5in]{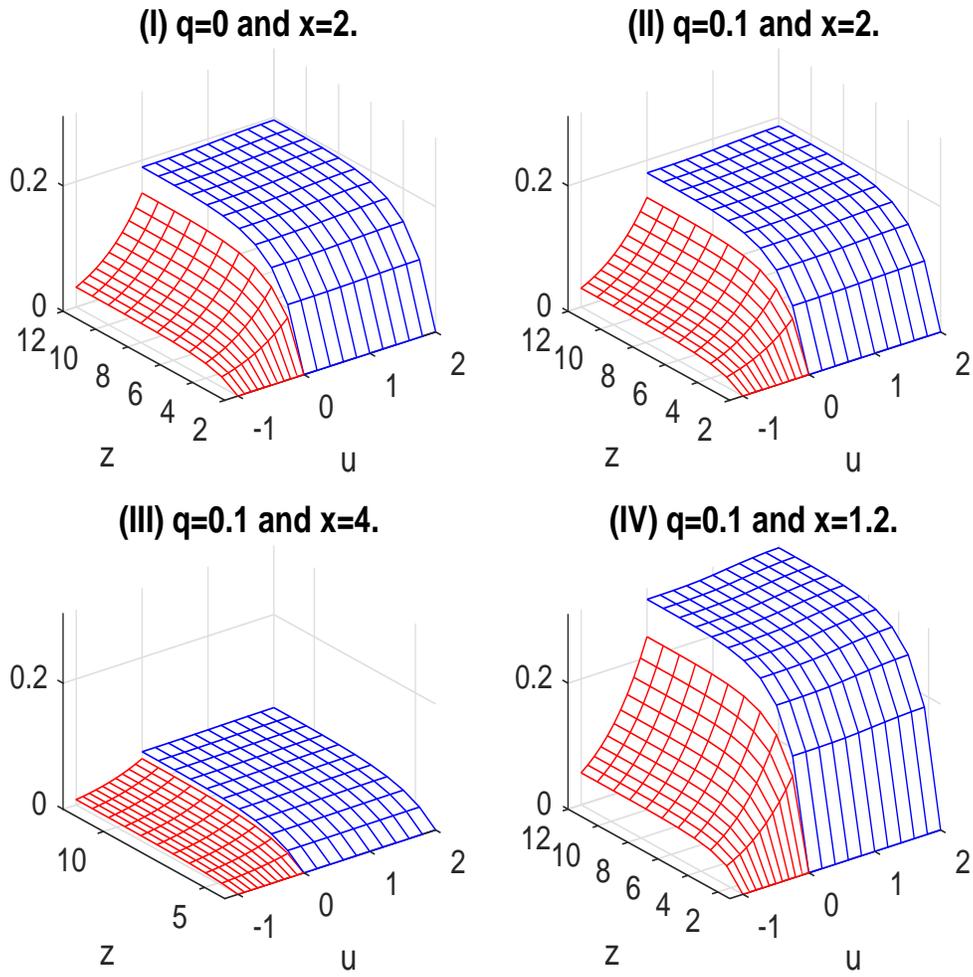}
\caption{The joint distribution of the liquidation time, the surplus at liquidation and historical high of the surplus at liquidation.}
\end{figure}

When the value of $q$ increases from 0 to 0.1 as shown in Figure 4(II), the value of the joint distribution will be compressed. We can observe a little bit lower level of the surfaces in (II) comparing to (I). When the initial wealth increases from $2$ to $4$, the time of liquidation will be deferred, then the value of the joint distribution is even lower, which can be easily seen in (III). In (IV), the initial wealth $x$ decreases to $1.2\in (b,c)$, which allows the insurer stay in the state of {\it solvent} but is already lower than the safety barrier $c$, then the financial stress of the insurer is heavier, which may lead to a much earlier time of liquidation, then the value of the joint distribution is increased as observed in (IV). These figures indicate that our theoretical results are consistent with the intuitive understanding of the model. The impacts of the three-barrier model on the extended expected discounted penalty function are also well studied.


\section{Conclusions}

This paper uses a three-barrier system to imitate the real-world liquidation process of an insurance company. According to the surplus level of the insurer, the financial status is divided into three states: {\it solvent}, {\it insolvent} and {\it liquidated}. The main feature of the current system is to allow the insurer to reorganize or rehabilitate its financial structure to return to the safety barrier within a granted grace period, rather than a direct bankruptcy in a single-barrier model. To describe the underlying surplus process, we use spectrally negative L\'{e}vy processes, which have been taken as good candidates to model insurance risks. Based on the three barriers, we provide a rigorous definition of the time of liquidation ruin. By adopting the techniques of excursions in the fluctuation theory, we study the joint distribution of the time of liquidation, the surplus at liquidation and the historical high of the surplus until liquidation, which generalizes the known results on the classical expected discounted penalty function. The results have semi-explicit expressions in terms of the scale functions and the L\'{e}vy triplet associated with the L\'{e}vy process. By studying the special cases in our setting, our result show its good consistency with the existing literature. Besides, several numerical examples are provided to compare the Parisian ruin and liquidation ruin, the impacts of model parameters on the liquidation ruin probability, and the joint distribution of the liquidation time, the surplus at liquidation and the historical high of the surplus until liquidation, which are consistent with both the intuitive and theoretical explanations.

\end{document}